\documentclass{article}
\usepackage[utf8]{inputenc}

\usepackage{amsmath,amssymb,amsfonts,amsthm,float}

\usepackage[nottoc]{tocbibind}
\usepackage[main=british,british]{babel}
\usepackage[backend=biber, style=numeric,]{biblatex}
\theoremstyle{plain}
\newtheorem{thm}{Theorem}[section]
\newtheorem{lem}[thm]{Lemma}
\newtheorem{prop}[thm]{Proposition}
\newtheorem{cor}[thm]{Corollary}

\theoremstyle{definition}
\newtheorem{defn}[thm]{Definition}

\newtheorem{exmp}[thm]{Example}

\theoremstyle{remark}

\newtheorem*{note}{Note}

\usepackage{chngcntr}
\counterwithin{table}{section}
\counterwithin{figure}{section}
\counterwithin{equation}{section}

\usepackage{mathrsfs}

\usepackage{aligned-overset}

\usepackage{geometry}
\title{The MacWilliams Identity for the Skew Rank Metric}
\author{Izzy Friedlander\footnote{Department of Computer Science, Durham University, UK. \texttt{isobel.s.friedlander@durham.ac.uk}}, Thanasis Bouganis\footnote{Department of Mathematical Sciences, Durham University, UK. \texttt{athanasios.bouganis@durham.ac.uk}}, Maximilien Gadouleau\footnote{Department of Computer Science, Durham University, UK. \texttt{m.r.gadouleau@durham.ac.uk}} % \footnote{Corresponding author}
}
\date{\today}

\begin{document}

\maketitle

% \newpage
\thispagestyle{plain}

\begin{abstract}
The weight distribution of an error correcting code is a crucial statistic in determining it's performance. One key tool for relating the weight of a code to that of it's dual is the MacWilliams Identity, first developed for the Hamming metric. This identity has two forms: one is a functional transformation of the weight enumerators, while the other is a direct relation of the weight distributions via (generalised) Krawtchouk polynomials. The functional transformation form can in particular be used to derive important moment identities for the weight distribution of codes. In this paper, we focus on codes in the skew rank metric. In these codes, the codewords are skew-symmetric matrices, and the distance between two matrices is the skew rank metric, which is half the rank of their difference. This paper develops a $q$-analog MacWilliams Identity in the form of a functional transformation for codes based on skew-symmetric matrices under their associated skew rank metric. The method introduces a skew-$q$ algebra and uses generalised Krawtchouk polynomials. Based on this new MacWilliams Identity, we then derive several moments of the skew rank distribution for these codes. 
\end{abstract}


% \newpage
% \tableofcontents
% \newpage
% \newpage
\section{Introduction}
Error correcting codes have been extensively and successfully used both for encoding of data in communications and storage \cite{TheoryofError}\cite{DensoADC7} and for code based cryptography \cite{McEliece11}. Besides the very important real life applications, they have also some very deep connections to other mathematical objects such as lattices and modular forms \cite{ConwayandSloane}.

Linear codes is an important subclass of error-correcting codes which has been extensively studied and used in practise, since the vector space structure can be used, among other things, for efficient encoding and decoding algorithms. The first, and perhaps a natural metric to consider for many applications, is the Hamming metric  

\cite{TheoryofError}\cite{HammingRW5} but others have since followed including, perhaps most notably, the rank metric explored by Delsarte \cite{DelsarteBlinear} and Gabidulin \cite{GabidulinTheory}. This has since been applied in many practical fields, such as error control in data storage \cite{RothRM}, space-time coding \cite{Tarokh}, and error control for network coding \cite{SilvaKoetter}.

An important statistic of a linear code is its weight distribution which encodes in a form of a homogeneous polynomial (weight enumerator) in two variables the number of codewords of various weight. This statistic has been studied extensively and had been used to obtain important bounds on the existence of codes. 
Among the tools that have been derived to analyse the weight distribution of a code is the widely used MacWilliams identity originally identified for the Hamming metric \cite{TheoryofError}. The identity relates the weight distribution of a code to that of its dual under the operation of an inner product defined on the space. There are various forms of the MacWilliams identity with each one having its own merits. For example the form stated in \cite{TheoryofError}, and extended in this paper here, can be used in combination with invariant theory to study self-dual codes. Here we have in mind the famous Gleason theorem and its consequences \cite{WeightPolys10}. 

Codes with the rank metric have been studied in depth by Delsarte \cite{DelsarteBlinear}\cite{DelsarteAlternating} and Gabidulin \cite{GabidulinTheory}. Delsarte developed a version of the MacWilliams identity using the theory of association schemes and subsequently Gadouleau and Yan \cite{gadouleau2008macwilliams} derived an alternative $q$-analog form of the identity using character theory and the Hadamard Transform \cite{TheoryofError}. Both theories can be compared through the associated generalised Krawtchouk Polynomials \cite{DelsarteBlinear}. 

Specifically, the identity developed in \cite{gadouleau2008macwilliams} is in the form of a functional transformation and is, as a result, both computationally effective and remarkably similar in form as a $q$-analog of the original MacWilliams identity for the Hamming metric. In this paper a new $q$-analog MacWilliams identity is derived for codes based on alternating bilinear forms (and their corresponding skew-symmetric matrices) which similarly has the form of a functional transform. The method builds on the work of Gadouleau and Yan to construct the components and structure of the identity but uses the theory of generalised Krawtchouk polynomials to complete the proof. In doing so, a new explicit form of the generalised Krawtchouk polynomials has been established.

The new MacWilliams identity then allows us to derive several results on the weight distribution of codes. Notably, we derive $q$-analogs of the relations between the binomial moments of the weight distribution of a linear code and that of its dual. In particular, depending on the minimum distance of a dual, we determine the moments of the weight distribution exactly. As a final application of our results, we then give an alternate proof of the weight distribution of optimal codes given in \cite{DelsarteAlternating}.

The rest of this paper is structured as follows: In Section \ref{section:prelims} the necessary definitions and properties are introduced and some important identities are derived. Section \ref{section:skew-q-transformetc} defines the skew-$q$-product, skew-$q$-power and skew-$q$-transform for homogeneous polynomials. In particular, the powers of two specific key polynomials are found and related to the weight enumerators of alternating bilinear forms of any order. In Section \ref{section:MacWilliamsIdentity} a new explicit form of the generalised Krawtchouk polynomials is established and is used to prove a $q$-analog of the MacWilliams identity for the skew rank metric as a functional transform. Section \ref{section:derivatives} introduces two derivatives for real valued functions of a variable and derives some results for homogeneous polynomials including the two key polynomials explored in Section \ref{section:skew-q-transformetc}. The derivatives are then used in Section \ref{section:moments} to identify moments of the skew rank distribution for linear codes based on skew-symmetric matrices. 

The results presented in this paper are included in \cite{IzzyThesis}, and they open clearly the possibility of obtaining similar results for other association schemes. Already in \cite{IzzyThesis} the case of Hermitian matrices is investigated and it is also natural to ask whether this may be extended to more general schemes such as translation association schemes. The crucial question here is whether one can define the analogue of the $q$-product in a general setting such that the MacWilliams identity can be stated in a functional form, as the one in \cite{gadouleau2008macwilliams} and the one obtained here.

\section{Preliminaries}\label{section:prelims}

\subsection{Skew-Symmetric Matrices}

\begin{defn}
Let $\boldsymbol{A}$ be a matrix of size $t\times t$ with entries in a finite field $\mathbb{F}_q$ where $q$ is a prime power. Then
$\boldsymbol{A}=(a_{ij})$ is called a \textit{\textbf{skew-symmetric}} matrix, if $\boldsymbol{A}^T=-\boldsymbol{A}$.
\end{defn}
The set of these skew-symmetric matrices is denoted $\mathscr{A}_{q,t}$ and the order of the matrix is $t$.

Each skew-symmetric matrix, $\boldsymbol{A}$, can be associated with a corresponding alternating bilinear form, which is a map
\begin{equation}
    \boldsymbol{A}:~ V\times V\rightarrow \mathbb{F}_q
\end{equation}
where $V$ is a $t$-dimensional vector space over $\mathbb{F}_q$ with fixed basis $\left\{\boldsymbol{e}_1,\boldsymbol{e}_2,\ldots,\boldsymbol{e}_t\right\}$ \cite{DelsarteAlternating} and
\begin{equation}
\boldsymbol{A}\left(\boldsymbol{e}_i,\boldsymbol{e}_j\right)=a_{ij}.
\end{equation}

The set of these bilinear forms is denoted $\mathbb{B}(t,q)$. There is a one to one correspondence between $\mathscr{A}_{q,t}$ and $\mathbb{B}(t,q)$.
\begin{thm}\label{vectorspaceproof}
$\mathscr{A}_{q,t}$ is a $t\choose 2$-dimensional vector space over $\mathbb{F}_q$.
\end{thm}

\begin{proof}
The proof of Theorem \ref{vectorspaceproof} is trivial and hence omitted.
\end{proof}

For $\mathscr{A}_{q,t}$ we define the parameters
\begin{equation}
n=\left\lfloor\frac{t}{2}\right\rfloor, ~m=\frac{t(t-1)}{2n}
\end{equation}
where $n$ is the maximum skew rank of $\boldsymbol{A}\in\mathscr{A}_{q,t}$ and $m$ is $t$ or $t-1$ depending if $t$ is odd or even.
We also follow the convention that empty product is taken to be $1$ and the empty sum is taken to be $0$.

\subsection{Properties of Skew-Symmetric Matrices}

An alternative way of defining a skew-symmetric matrix is as follows:

\begin{defn}[{\cite{AAA}}]
A matrix, $\boldsymbol{A}$ is skew-symmetric if and only if for any vector $\boldsymbol{x}$, $\boldsymbol{x}\boldsymbol{A}\boldsymbol{x}^T=0$.

\end{defn}

\begin{defn}

Two matrices $\boldsymbol{A}$ and $\boldsymbol{B}$ in $\mathscr{A}_{q,t}$ are said to be \textit{\textbf{congruent}} if there exists a non-singular $t\times t$ matrix $\boldsymbol{P}$ over $\mathbb{F}_q$ such that $\boldsymbol{B} = \boldsymbol{PAP}^T$.
\end{defn}

The following properties of skew-symmetric matrices are proved in \cite{AAA}.

\begin{enumerate}
    \item Two skew-symmetric matrices are congruent if and only if they have the same (column) rank.
    \item The rank of a skew-symmetric matrix is even.
    \item If the rank of a skew-symmetric matrix, $\boldsymbol{A}$ is $2s$ with $0 \leq s \leq n$, say, then $\boldsymbol{A}$ is congruent to the matrix

    $$\begin{pmatrix}
    E_2 &  &  &  &  \\ 
     & E_2 &  &  &  \\
     & & \ddots & &\\
     & & & E_2 & \\
     &  & &  & \mathcal{O}_{t-2s}
    \end{pmatrix}$$
    
where $E_2 = \begin{pmatrix} 0 & 1 \\ -1 & 0\end{pmatrix}$ and $\mathcal{O}_{t-2s}$ is the zero matrix of order $t-2s$. We will denote this matrix as diag$\left\{E_2,E_2,\ldots,E_2,\mathcal{O}_{t-2s}\right\}$, and call it the \textbf{\textit{canonical form of}} $\boldsymbol{A}$.
\end{enumerate}

\subsection{The Skew Rank of a Skew-Symmetric Matrix}\label{Sectionskewrank}

\begin{defn}
For all $\boldsymbol{A} \in \mathscr{A}_{q,t}$ with column rank $2s$ we define the \textit{\textbf{skew rank}} of  $\boldsymbol{A}$, $SR(\boldsymbol{A})$, to be $s$.\\

For all $\boldsymbol{A},\boldsymbol{B}\in \mathscr{A}_{q,t}$, we define the \textit{\textbf{skew rank distance}} to be 
\begin{equation}
    d_{SR}(\boldsymbol{A},\boldsymbol{B})=SR(\boldsymbol{A}-\boldsymbol{B}).
\end{equation}
It is easily verified that $d_{SR}$ is a metric over $\mathscr{A}_{q,t}$ since $SR(\boldsymbol{A}-\boldsymbol{B})$ is the rank metric \cite{GabidulinTheory} \cite{gadouleau2008macwilliams} divided by $2$ and we will call it the \textit{\textbf{skew rank metric}}.
\end{defn}

\subsection{Codes based on Subspaces of Skew-Symmetric Matrices}

Any subspace of $\mathscr{A}_{q,t}$ can be considered as an $\mathbb{F}_q$-linear code, $\mathscr{C}$, with each matrix of skew rank $s$ in $\mathscr{C}$ representing a codeword of weight $s$ and with the distance metric being the skew rank metric defined in Section \ref{Sectionskewrank}.

The \textbf{\textit{minimum skew rank distance}} of such a code $\mathscr{C}$, denoted as $d_{SR}(\mathscr{C})$, is simply the minimum skew rank distance over all possible pairs of distinct codewords in $\mathscr{C}$. When there is no ambiguity about $\mathscr{C}$, we denote the minimum skew rank distance as $d_{SR}$.\\

It can be shown that \cite[p.33]{DelsarteAlternating} the cardinality $|\mathscr{C}|$ of a code $\mathscr{C}$ over $\mathbb{F}_q$ based on $t \times t$ skew-symmetric matrices and minimum skew rank distance $d_{SR}$ satisfies 
\begin{equation}\label{SinglebuondSkew}
    |\mathscr{C}|\leq q^{m(n-d_{SR}+1)}
\end{equation}

In this paper, we call the bound in \eqref{SinglebuondSkew} the Singleton Bound for codes with the Skew Rank Metric.
Codes that attain the Singleton bound are referred to as Maximal Codes or Maximum Skew Rank Distance (MSRD) codes. 

\begin{defn}\label{def:skewrankweightenumerator}
For all $\boldsymbol{A}\in\mathscr{A}_{q,t}$ with skew rank weight $s$, the \textit{\textbf{skew rank weight function}} of $\boldsymbol{A}$ is defined as the homogeneous polynomial
\begin{equation}
f_{SR}(\boldsymbol{A})=Y^{s}X^{n-s}. 
\end{equation}
Let $\mathscr{C} \subseteq \mathscr{A}_{q,t}$ be a code. Suppose there are $c_i$ codewords in $\mathscr{C}$ with skew rank weight $i$ for $0\leq i\leq n$. Then the \textbf{\textit{skew rank weight enumerator}} of $\mathscr{C}$, denoted as $W_\mathscr{C}^{SR}(X,Y)$ is defined to be
\begin{equation}\label{defn:skewrankweightenumerator}
    W_\mathscr{C}^{SR}(X,Y) = \sum_{\boldsymbol{A}\in\mathscr{C}}f_{SR}(\boldsymbol{A})=\sum_{i=0}^n c_iY^{i}X^{n-i}.
\end{equation}

The $(n+1)$-tuple, $\boldsymbol{c}=(c_0,\ldots,c_n)$ of coefficients of the weight enumerator, is called the \textbf{\textit{weight distribution}} of the code $\mathscr{C}$.
\end{defn}

\begin{exmp}\label{example:numberofmatricesexplicit}
An example of such a code with $q = 3$ and $t = 4$ is where $\mathscr{C}$ is the set of skew-symmetric matrices, $\boldsymbol{A}=(a_{ij})$ with $1 \leq i,j \leq 4$, such that;
\begin{equation}
    \begin{cases}
    a_{1j} \in \mathbb{F}_q, ~j > 1\\
    a_{2j} = 0 \text{ for }i < j\\
    a_{34}\in\mathbb{F}_q
    \end{cases}
\end{equation}
There are $81$ matrices (codewords) in this code. The only codeword of skew rank $0$ is the all-zero matrix. It is easily seen that a codeword has skew rank $2$ if and only if $a_{12}$ and $a_{34}$ are both nonzero. Therefore, there are exactly $36$ codewords of skew rank $2$, and consequently exactly $44$ codewords of skew rank $1$. Thus, the skew rank weight enumerator of the code is $X^2+44XY+36Y^2$. \\

\end{exmp}

\subsection{Counting the number of Skew-Symmetric matrices of a given size}

Multiple ways of describing the number of skew-symmetric matrices have been developed by various authors such as \cite[Proposition 2.1, p627]{StantonChevally}, \cite[Theorem 3, p155]{MacWilliamsOrthogonal}, \cite[Theorem 2, p437]{TheoryofError} and \cite{DelsarteAlternating}. The following is (for the purpose of this paper) in the best format.

\begin{thm}[{\cite[Theorem 3, p24]{Carlitz1954}}]\label{thm:countingMatrices}
    The number of skew symmetric matrices of order $t$ and skew rank $s$ is
\begin{equation}\label{CarlitzCountingMatricesSkew}
\xi_{t,s} =
\begin{cases}
         q^{2\sigma_s}\times\dfrac{\displaystyle\prod_{i=0}^{2s-1}\left(q^{t-i}-1\right)}{\displaystyle\prod_{i=1}^s\left(q^{2i}-1\right)} & \text{if} ~0\leq s\leq n,\\
        0 & \text{otherwise}.
    \end{cases}
\end{equation}
\end{thm}

\begin{defn}
We also define the \textbf{\textit{skew rank weight enumerator}} of $\mathscr{A}_{q,t}$ to be
\begin{equation}
    \Omega_t= \sum_{i=0}^n \xi_{t,i}Y^{i}X^{n-i}.
\end{equation}

\end{defn}

\begin{exmp}
For $t=4$ and $q=3$ the skew rank weight enumerator is
\begin{align}
    X^2+\left(3^2+1\right)\left(3^3-1\right)XY+3^2\left(3^3-1\right)(3-1)Y^2 & = 
    X^2+\left(10\right)\left(26\right)XY+9\left(26\right)(2)Y^2\\
    & = X^2 + 260XY + 468Y^2.
\end{align}
\end{exmp}

\subsection{Inner product of two Skew-Symmetric matrices}
We define an \textbf{inner product} on $\mathscr{A}_{q,t}$ by
\begin{equation}
    (\boldsymbol{A},\boldsymbol{B})\mapsto\langle \boldsymbol{A},\boldsymbol{B} \rangle =Tr\left(\boldsymbol{A}^T\boldsymbol{B}\right)
\end{equation}
where $Tr(\boldsymbol{A})$ means the trace of $\boldsymbol{A}$.

\begin{defn}
The \textbf{\textit{dual}} of a code, $\mathscr{C}$, denoted by $\mathscr{C}^\perp$ is defined as
\begin{equation}
    \mathscr{C}^\perp = \big\{ \boldsymbol{A}\in\mathscr{A}_{q,t}~ \vert
\left\langle \boldsymbol{A},\boldsymbol{B}\right\rangle =0~\forall~ \boldsymbol{B} \in \mathscr{C}\big\}.\end{equation}
\end{defn}

\begin{thm}[{\cite[Theorem 5]{DelsarteAlternating}}]
A code $\mathscr{C}\subseteq\mathscr{A}_{q,t}$ with minimum skew rank distance $d_{SR}$ is MSRD if and only if its dual $\mathscr{C}^\perp$ is also MSRD with minimum skew rank distance $d_{SR}'=n-d_{SR}+2$.
\end{thm}

\subsection{Skew-$q$-nary Gaussian Coefficients and other useful identities}\label{subsection:b-naryidentities}

In establishing the results later in this paper we have used some identities to simplify the notation and algebra.

\begin{defn}\label{def:skewqnarygaussiancoeffs}
For any real number $q\neq 1$, $k\in\mathbb{Z}^+$ and $x\in \mathbb{R}$ (usually an integer), we define the \textbf{\textit{Skew-$q$-nary Gaussian Coefficients}} \cite{DelsarteAlternating}, $\left[\begin{matrix} x \\ k\end{matrix}\right]$, to be
\begin{equation}
    \left[\begin{matrix} x \\ k\end{matrix}\right] = \prod_{i=0}^{k-1}\frac{q^{2x}-q^{2i}}{q^{2k}-q^{2i}}
\end{equation}
with 
\begin{equation}
    \left[\begin{matrix} x \\ 0\end{matrix}\right] =1.
\end{equation}
\end{defn}
If $x\in\mathbb{Z}^+$ then these skew-$q$-nary Gaussian coefficients count the number of $k$-dimensional subspaces of an $x$-dimensional vector space over $\mathbb{F}_{q^2}$ \cite[p3]{GabidulinTheory}.
Here are some identities relating to the Skew-$q$-nary Gaussian coefficients that are useful from \cite{DelsarteAlternating}:
\begin{align}
    \left[\begin{matrix} x \\ k\end{matrix}\right] & = \left[\begin{matrix} x \\ x-k\end{matrix}\right]\\
    \left[\begin{matrix} x \\ i\end{matrix}\right]\left[\begin{matrix} x-i \\ k\end{matrix}\right] & = \left[\begin{matrix} x \\ k\end{matrix}\right]\left[\begin{matrix} x-k \\ i\end{matrix}\right]\label{equation:gaussianswapplaces}\\
    \prod_{i=0}^{x-1}\left(y-q^{2i}\right)  & = \sum_{k=0}^x (-1)^{x-k}q^{2{x-k\choose 2}}\left[\begin{matrix} x \\ k\end{matrix}\right]y^k\\
    \sum_{k=0}^x\left[\begin{matrix} x \\ k\end{matrix}\right]\prod_{i=0}^{k-1}\left(y-q^{2i}\right)  & = y^x\label{equation:producttosumgauss}\\
    \sum_{k=i}^j (-1)^{k-i}q^{2{k-i\choose 2}}\left[\begin{matrix} k \\ i\end{matrix}\right]\left[\begin{matrix} j \\ k\end{matrix}\right] & = \delta_{ij}.\label{equation:deltaijbs}
\end{align}

The following additional identities are proven in \cite{andrews_1984}.
\begin{align}
    \left[\begin{matrix} x \\ k\end{matrix}\right] & = \left[\begin{matrix} x-1 \\ k\end{matrix}\right] + q^{2(x-k)}\left[\begin{matrix} x-1 \\ k-1\end{matrix}\right]\label{equation:thing5}\\
    & = \left[\begin{matrix} x-1 \\ k-1\end{matrix}\right]+ q^{2k}\left[\begin{matrix} x-1 \\ k\end{matrix}\right]\label{equation:thing1}\\
    & = \dfrac{q^{2(x-k+1)}-1}{q^{2k}-1}\left[\begin{matrix} x \\ k-1\end{matrix}\right]\label{equation:thing3}\\
    & = \dfrac{q^{2x}-1}{q^{2(x-k)}-1}\left[\begin{matrix} x-1 \\ k\end{matrix}\right].\label{equation:thing2}
\end{align}
Combining \eqref{equation:thing1} and \eqref{equation:thing2} gives,
\begin{align}
    \begin{bmatrix}x-1\\k-1\end{bmatrix} \overset{\eqref{equation:thing1}}&{=}  \begin{bmatrix}x\\k\end{bmatrix}-q^{2k}\begin{bmatrix}x-1\\k\end{bmatrix}\\
    \overset{\eqref{equation:thing2}}&{=}  \begin{bmatrix}x\\k\end{bmatrix}-\dfrac{q^{2k}\left(q^{2(x-k)}-1\right)}{q^{2x}-1}\begin{bmatrix}x\\k\end{bmatrix}\\
    & = \begin{bmatrix}x\\k\end{bmatrix}\left(1-\dfrac{q^{2k}\left(q^{2(x-k)}-1\right)}{q^{2x}-1}\right)\\
    & = \dfrac{q^{2k}-1}{q^{2x}-1}\begin{bmatrix}x\\k\end{bmatrix}.\label{equation:thing4}
\end{align}
\begin{defn}
We define the \textit{\textbf{Skew-$q$-nary Gamma function}} for $x\in\mathbb{R}$, $k\in\mathbb{Z}$ to be
\begin{equation}
\gamma(x,k)=\displaystyle\prod_{i=0}^{k-1}\left(q^x-q^{2i}\right).
\end{equation}

\end{defn}

Theorem \ref{thm:countingMatrices} can then be rewritten as
\begin{equation}\label{equation:numbermatricesgamma}
    \xi_{t,k}= \left[\begin{matrix} n \\ k\end{matrix}\right]\gamma(m,k).
\end{equation}

\begin{lem}\label{lemma:Gammaidentites}
We have the following identities for the skew-$q$-nary Gamma function:
\begin{equation}\gamma(x,k) = q^{k(k-1)}\displaystyle\prod_{i=0}^{k-1}\left(q^{x-2i}-1\right),\end{equation}
    \begin{equation}{\dfrac{\gamma(2x,k)}{\gamma(2k,k)} = \left[\begin{matrix} x \\ k\end{matrix}\right]=\frac{\prod_{i=0}^{k-1}\left(q^{2x-2i}-1\right)}{\prod_{i=1}^k \left(q^{2i}-1\right)},}\end{equation} \begin{equation}\label{equation:gamma2step} \gamma(x+2,k+1) = \left(q^{x+2}-1\right)q^{2k}\gamma(x,k),\end{equation}
    \item \begin{equation}\label{equation:gamma1step}\gamma(x,k+1) = \left(q^x-q^{2k}\right)\gamma(x,k).\end{equation}
\end{lem}

\begin{proof}~\\
$(1)$
\begin{align}
   \gamma(x,k) & = \prod_{i=0}^{k-1}\left( q^x - q^{2i}\right)\\
   &= \left(\prod_{i=0}^{k-1}q^{2i}\right)\prod_{i=0}^{k-1}\left(q^{x-2i}-1\right)\\
   & = q^{k(k-1)}\prod_{i=0}^{k-1}\left(q^{x-2i}-1\right).
\end{align}

$(2)$
\begin{equation}
    \left[\begin{matrix} x \\ k\end{matrix}\right]=\dfrac{\displaystyle\prod_{i=0}^{k-1}\left(q^{2x}-q^{2i}\right)}{\displaystyle\prod_{i=0}^{k-1}\left(q^{2k}-q^{2i}\right)} = \dfrac{\gamma(2x,k)}{\gamma(2k,k)}=\frac{\displaystyle\prod_{i=0}^{k-1}\left(q^{2x-2i}-1\right)}{\displaystyle\prod_{i=1}^{k}\left(q^{2i}-1\right)}.
\end{equation}

$(3)$
\begin{align}
    \gamma(x+2,k+1) & = \left( q^{x+2}-1\right)\prod_{i=1}^k \left(q^{x+2}-q^{2i}\right)\\
    & = \left(q^{x+2}-1\right)q^{2k}\displaystyle\prod_{i=0}^{k-1} \left(q^x-q^{2i}\right)\\
    & = \left(q^{x+2}-1\right)q^{2k}\gamma(x,k).
    \end{align}
$(4)$
\begin{align}
    \gamma(x,k+1) & = \prod_{i=0}^k\left(q^x-q^{2i}\right)\\
    & = \left(q^x-q^{2k}\right)\prod_{i=0}^{k-1}\left(q^x-q^{2i}\right)\\
    & = \left(q^x-q^{2k}\right)\gamma(x,k).
\end{align}
\end{proof}

\begin{defn}\label{defn:skewqbetafunction}
We also define a \textbf{\textit{Skew-$q$-nary Beta function}} for $x\in\mathbb{R}$, $k\geq 0$ as 
\begin{equation}\label{equation:skewqbetafunction}
    \beta(x,k) =
             \displaystyle\prod_{i=0}^{k-1}\begin{bmatrix}x-i\\1\end{bmatrix}.
\end{equation}
These are closely related to Skew-$q$-Gaussian Coefficients. 
\end{defn}

We also define $\sigma_i=\dfrac{i(i-1)}{2}$ for $i\geq0.$

\begin{lem}\label{lemma:betabmanipulation}
We have for all $x\in\mathbb{R}$, $k\geq 0$,
\begin{equation}\label{equation:betabstartdifferent}
\beta(x,k) = \begin{bmatrix}x\\k\end{bmatrix}\beta(k,k)
\end{equation}
and
\begin{equation}\label{equation:betabstartsame}
\beta(x,x) = \begin{bmatrix}x\\k\end{bmatrix}\beta(k,k)\beta(x-k,x-k).
\end{equation}
\end{lem}
\begin{proof}
We have
\begin{align}
    \beta(x,k) = \prod_{i=0}^{k-1}\begin{bmatrix}x-i\\1\end{bmatrix} & = \prod_{i=0}^{k-1}\frac{q^{2(x-i)}-1}{q^2-1}\\
    & = \prod_{i=0}^{k-1}\frac{\left(q^{2(x-i)}-1\right)\left(q^{2(k-i)}-1\right)}{\left(q^{2(k-i)}-1\right)(q^2-1)}\\
    & = \prod_{i=0}^{k-1}\left(\frac{q^{2x}-q^{2i}}{q^{2k}-q^{2i}}\right)\prod_{i=0}^{k-1}\left(\frac{q^{2(k-i)}-1}{q^2-1}\right)\\
    & = {\begin{bmatrix}x\\k\end{bmatrix}}\beta(k,k)
\end{align}
as required. Now we have

\begin{align}
    {\begin{bmatrix}x\\k\end{bmatrix}}\beta(k,k)\beta(x-k,x-k) & = \prod_{i=0}^{k-1}\left(\frac{q^{2x}-q^{2i}}{q^{2k}-q^{2i}}\right)\prod_{r=0}^{k-1}\left(\frac{q^{2(k-r)}-1}{q^2-1}\right)\prod_{s=0}^{x-k-1}\left(\frac{q^{2(x-k-s)}-1}{q^2-1}\right)\\
    & = \prod_{i=0}^{x-1}\frac{q^{2(x-i)}-1}{q^2-1}\\
    & = \beta(x,x)
\end{align}
as required.
\end{proof}
\section{The Skew-$q$-Product and Skew-$q$-Transform}\label{section:skew-q-transformetc}

The weight enumerators of any linear code $\mathscr{C}\subseteq \mathscr{A}_{q,t}$ are homogeneous polynomials. We introduce an operation, the Skew-$q$-Product, on homogeneous polynomials that will help to express the relation between the weight enumerator of a code and that of it's dual.

\subsection{The Skew-$q$-product, Skew-$q$-power and the Skew-$q$-transform}
\begin{defn}\label{q-proddefn}
Let 
\begin{align}
    a(X,Y;\lambda) & = \sum_{i=0}^r a_i(\lambda)Y^iX^{r-i},\\
    b(X,Y;\lambda) & = \sum_{i=0}^s b_i(\lambda)Y^iX^{s-i},
\end{align}
be two homogeneous polynomials in $X$ and $Y$, of degrees $r$ and $s$ respectively, and coefficients $a_i(\lambda)$ and $b_i(\lambda)$ respectively, which are real functions of $\lambda$ and are 0 unless otherwise specified, for example $b_i(\lambda)=0$ if $i\notin\{0,1,\ldots, s\}$.
The \textit{\textbf{skew-$\boldsymbol{q}$-product}}, $\ast$, of $a(X,Y;\lambda)$, of order $r$, and $b(X,Y;\lambda)$, is defined as

\begin{equation}\label{qskewproduct}
    \begin{split}
        c(X,Y;\lambda) & = a(X,Y;\lambda) \ast b(X,Y;\lambda) \\
        & = \sum_{u=0}^{r+s} c_u(\lambda)Y^uX^{r+s-u}
    \end{split}
\end{equation}
with
\begin{equation}
    c_u(\lambda) = \sum_{i=0}^u q^{2is} a_i(\lambda)b_{u-i}(\lambda-2i).
\end{equation}
\end{defn}

\begin{defn}
As in \cite{gadouleau2008macwilliams}, the \textbf{\textit{skew-$\boldsymbol{q}$-power}} is defined by
\begin{equation}
\begin{cases}
    a^{[0]}(X,Y;\lambda) =1,\\
    a^{[1]}(X,Y;\lambda) =a(X,Y;\lambda),\\
    a^{[k]}(X,Y;\lambda) = a(X,Y;\lambda) \ast a^{[k-1]}(X,Y;\lambda) &\text{for } k\geq 2.
    \end{cases}
\end{equation}
\end{defn}

\begin{defn}[{\cite[Definition 4]{gadouleau2008macwilliams}}]\label{skew-q-transform}
Let $a(X,Y;\lambda)=\displaystyle\sum_{i=0}^r a_i(\lambda)Y^iX^{r-i}$. We define the \textbf{\textit{skew-$\boldsymbol{q}$-transform}} to be the homogeneous polynomial
\begin{equation}
    \overline{a}(X,Y;\lambda) = \sum_{i=0}^r a_i(\lambda)Y^{[i]}\ast X^{[r-i]}
\end{equation}
where $Y^{[i]}$ is the $i^{th}$ skew-$q$-power of the homogeneous polynomial $a(X,Y;\lambda)=Y$ and $X^{[r-i]}$ is the $r-i^{th}$ skew-$q$-power of the homogeneous polynomial $a(X,Y;\lambda)=X$.
\end{defn}
 
\subsection{Using the Skew-$q$-Product to identify the Rank Weight Enumerator of Skew-Symmetric Matrices}\label{subsection:usingskewqproduct}

In the theory that follows, relating the weight enumerator of a code to it's dual, then we consider the following polynomial. Let
\begin{equation}
    \mu(X,Y;\lambda) = X+\left(q^\lambda-1\right)Y\label{equationforb}.
\end{equation}

\begin{thm}\label{bformula}
If $\mu(X,Y;\lambda)$ is as defined above, then
\begin{equation}
    \mu^{[k]}(X,Y;\lambda) = \sum_{u=0}^{k} \mu_u(\lambda,k)Y^{u}X^{k-u} \quad \text{ for } k\geq 1,
\end{equation}
where
\begin{equation}
    \mu_u(\lambda,k) = {\left[\begin{matrix} k \\ u\end{matrix}\right]}\gamma(\lambda,u).
\end{equation}
Specifically, the weight enumerators for $\mathscr{A}_{q,t}$, the set of skew-symmetric matrices of size $t\geq 1$, denoted by $\Omega_t$, is given by,
\begin{equation}
    \Omega_{t} = \mu^{[n]}(X,Y;m)
\end{equation}
where $m=\frac{t(t-1)}{2n}$. In other words, the skew-$q$-powers of $\mu(X,Y;m)$ provide an explicit form for the weight enumerator of $\mathscr{A}_{q,t}$, the set of skew-symmetric matrices of order $t$.
\end{thm}

\begin{proof}
The proof follows the method of induction.\\
Consider $k=1$.
\begin{equation}
    \mu_0(\lambda,1) = {\begin{bmatrix}1\\0\end{bmatrix}}\gamma(\lambda,0)=1,
\end{equation}
\begin{equation}
    \mu_1(\lambda,1) = {\begin{bmatrix}1\\1\end{bmatrix}}\gamma(\lambda,1)=\left(q^\lambda-1\right).
\end{equation}

So,
\begin{align}
    \mu^{[k]} = \mu^{[1]} & = \mu\\
    & = X +\left(q^\lambda-1\right)Y\\
    & = \mu_0(m,1)Y^0X^1 + \mu_1(\lambda,1)Y^1X^0\\
    & = \sum_{u=0}^k \mu_u(\lambda,k)Y^{u}X^{k-u}
\end{align}
as required. Now assume the theorem is true for $k\geq 1$. 
\begin{align}
    \mu^{[k+1]} & = \mu \ast \mu^{[k]}\\
    & = \left( X+\left(q^\lambda-1\right)Y\right)\ast \left(\sum_{u=0}^k \mu_u(\lambda,k) Y^{u}X^{k-u}\right)\\
    & = \left(\sum_{u=0}^1 \mu_u(\lambda,1)Y^{u}X^{1-u}\right)\ast\left(\sum_{u=0}^k \mu_u(\lambda,k) Y^{u}X^{k-u}\right)\\
    & = \sum_{i=0}^{k+1}f_i(\lambda)Y^{i}X^{k+1-i}
\end{align}
where \begin{equation}
    f_i(\lambda) = \sum_{j=0}^i q^{2jk}\mu_j(\lambda,1)\mu_{i-j}(\lambda-2j,k)
\end{equation}
by definition of the skew-$q$-product.\\
If $i=0$,
\begin{equation}
    f_0(\lambda)=q^0\mu_0(\lambda,1)\mu_0(\lambda,k)=1,
\end{equation}
and if $i\geq1$,
\begin{align}
    f_i(\lambda) & = (1)\mu_0(\lambda,1)\mu_i(\lambda,k) + q^{2k}\mu_1(\lambda,1)\mu_{i-1}(\lambda-2,k)\\
    & = {\begin{bmatrix} k\\i\end{bmatrix}}\gamma(\lambda,i)+q^{2k}\left(q^\lambda-1\right){\begin{bmatrix}k\\i-1\end{bmatrix}}\gamma(\lambda-2,i-1).
\end{align}
Now,
\begin{equation}
    \left(q^\lambda-1\right)\gamma(\lambda-2,i-1) \overset{\eqref{equation:gamma2step}}{=} q^{-2(i-1)}\gamma(\lambda,i)
\end{equation}
by rearranging Lemma \ref{lemma:Gammaidentites}. We also have
\begin{equation}
    {\begin{bmatrix}k\\i\end{bmatrix}} \overset{\eqref{equation:thing2}}{=} \dfrac{q^{2(k-i+1)}-1}{q^{2(k+1)}-1}{\begin{bmatrix}k+1\\i\end{bmatrix}}
\end{equation}
and
\begin{equation}
    {\begin{bmatrix}k\\i-1\end{bmatrix}} \overset{\eqref{equation:thing4}}{=} \dfrac{q^{2i}-1}{q^{2(k+1)}-1}{\begin{bmatrix}k+1\\i\end{bmatrix}}.
\end{equation}

Therefore, 
\begin{align}
    f_i(\lambda) & =\dfrac{q^{2(k+1-i)}-1}{q^{2(k+1)}-1}{\begin{bmatrix}k+1\\i\end{bmatrix}}\gamma(\lambda,i) + q^{2k}q^{-2(i-1)}\gamma(\lambda,i)\dfrac{q^{2i}-1}{q^{2(k+1)}-1}\\
    & = \dfrac{{\begin{bmatrix}k+1\\i\end{bmatrix}}\gamma(\lambda,i)}{q^{2(k+1)}-1}q^{2(k+1-i)}-1+q^{2(k+1)}-q^{2(k+1-i)}\begin{bmatrix}k+1\\i\end{bmatrix}\\
    & = {\begin{bmatrix}k+1\\i\end{bmatrix}}\gamma(\lambda,i)
\end{align}
as required. It follows immediately from Equation \eqref{equation:numbermatricesgamma} that $\mu_u(m,n)=\xi_{t,u}$. So $\mu^{[n]}(X,Y;m)=\Omega_t$.
\end{proof}

Now let $\nu(X,Y;\lambda)=X-Y.$

\begin{thm}\label{thed'slemma}
For all $k\geq 1$,
\begin{equation}\label{squaredminussquare}
    \nu^{[k]}(X,Y;\lambda) =\sum_{u=0}^k (-1)^u q^{u(u-1)}{\begin{bmatrix}k \\ u \end{bmatrix}}Y^{u}X^{k-u}.
\end{equation}
\end{thm}

\begin{proof}
We perform induction on $k$. It is easily checked that the theorem holds for $k=1$.

Now assume the theorem holds for $k\geq 1$.
Then
\begin{align}
    \nu^{[k+1]} & = \nu \ast \nu^{[k]}\\
    & = \left(X-Y\right) \ast\left(\sum_{u=0}^k(-1)^u q^{u(u-1)}{\begin{bmatrix}k\\u\end{bmatrix}}Y^{u}X^{k-u}\right)\\
    & = \sum_{i=0}^{k+1}g_i(\lambda) Y^{i}X^{k+1-i}.
\end{align}
We write $\nu=X-Y=\nu_0X^1Y^0+\nu_1X^0Y^1$ where $\nu_0=1$ and $\nu_1=-1$. Then,
\begin{align}
    g_i(\lambda) & = \sum_{j=0}^i q^{2jk}\nu_j(\lambda)\left\{(-1)^{i-j}q^{(i-j)(i-j-1)}{\begin{bmatrix}k\\i-j\end{bmatrix}}\right\}\\
    & = (-1)^iq^0\nu_0(\lambda)q^{i(i-1)}{\begin{bmatrix}k\\i\end{bmatrix}} + (-1)^{i-1}q^{2k}\nu_1(\lambda)q^{(i-1)(i-2)}{\begin{bmatrix}k\\i-1\end{bmatrix}}
\end{align}
but we have
\begin{equation}
    {\begin{bmatrix}k\\i\end{bmatrix}} \overset{\eqref{equation:thing3}}{=} \dfrac{q^{2(k-i+1)}-1}{q^{2(k+1)}-1}{\begin{bmatrix}k+1\\i\end{bmatrix}}
\end{equation}
and 
\begin{equation}
    {\begin{bmatrix}k\\i-1\end{bmatrix}} \overset{\eqref{equation:thing4}}{=} \dfrac{q^{2i}-1}{q^{2(k+1)}-1}{\begin{bmatrix}k+1\\i\end{bmatrix}}.
\end{equation}
So,
\begin{align}
    g_i(\lambda) & = (-1)^i q^{i(i-1)}\dfrac{q^{2(k-i+1)}-1}{q^{2(k+1)}-1}{\begin{bmatrix}k+1\\i\end{bmatrix}} + (-1)^iq^{2k}q^{i(i-1)}q^{-2(i-1)}\dfrac{q^{2i}-1}{q^{2(k+1)}-1}{\begin{bmatrix}k+1\\i\end{bmatrix}}\\
    & = \dfrac{(-1)^iq^{i(i-1)}}{q^{2(k+1)}-1}{\begin{bmatrix}k+1\\i\end{bmatrix}}\left\{q^{2(k-i+1)}-1+q^{2k-2i+2+2i}-q^{2k-2i+2}\right\}\\
    & = (-1)^iq^{i(i-1)}{\begin{bmatrix}k+1\\i\end{bmatrix}}
\end{align}
as required.

\end{proof}
\section{The MacWilliams Identity for the Skew Rank metric}\label{section:MacWilliamsIdentity}

In this section we introduce the Skew-$q$-Krawtchouk polynomials which we then prove are equal to the generalised Krawtchouk polynomials that are identified in \cite[(15)]{delsartereccurance}\cite[(A10)]{DelsarteBlinear} for the association schemes of alternating bilinear forms over $\mathbb{F}_q$. In this way a new $q$-analog of the MacWilliams Identity for dual subgroups (or codes) of alternating bilinear forms over $\mathbb{F}_q$ is presented and proven by comparison with a traditional form of the identity as given in \cite[Theorem 3]{DelsarteAlternating} and proved in \cite{delsarte1973algebraic} and \cite[(3.14)]{DelsarteBlinear}.

\subsection{Generalised Krawtchouk Polynomials}\label{subsection:generalisedKpolynomials}

We first recall the definition of the Krawtchouk polynomials in the setting of skew-symmetric matrices as in \cite{delsartereccurance}.

\begin{defn}\label{def:generalisedKpolynomial}
For any real number $b\geq 1$ and $c>\frac{1}{b}$ and for $x,k\in\left\{ 0,1,\ldots,y\right\}$ with $y\in\mathbb{Z}^+$ the \textbf{\textit{generalised Krawtchouk Polynomial}}, $P_k(x,y)$, is defined by
\begin{equation}
    P_k(x,y) = \sum_{j=0}^k{(-1)}^{k-j}\left(cb^y\right)^jb^{k-j\choose 2}{\begin{bmatrix}y-j\\y-k\end{bmatrix}_b}{\begin{bmatrix} y-x \\j \end{bmatrix}_b}\label{equation:generalisedKpolynomial}
\end{equation}
where we define the $b$-nary Gaussian Coeffients to be $\begin{bmatrix}x\\k
\end{bmatrix}_b=\displaystyle\prod_{i=0}^{k-1}\frac{b^x-b^i}{b^k-b^i}$ which has the same properties as the skew-$q$-nary Gaussian Coefficients (Definition \ref{def:skewqnarygaussiancoeffs}).
 Note that if $b=1$ these $P_k(x,y)$ are the usual Krawtchouk Polynomials as used, for example, in \cite{TheoryofError}.

\end{defn}

In this paper use is made of the recurrence relation below and it's family of solutions, generalised Krawtchouk Polynomials, as defined above. The recurrence relation, for $b\in\mathbb{R}^+$, $y\in\mathbb{Z}^+$ and $x,k\in\{0,1,\ldots,y\}$ is 
\begin{equation}\label{equation:recurrencerelation}
    P_{k+1}(x+1,y+1) = b^{k+1}P_{k+1}(x,y)-b^kP_{k}(x,y)
\end{equation}

and it's solutions are examined in \cite{delsartereccurance}.

The $P_k(x,y)$ are the only solutions to the recurrence relation \eqref{equation:recurrencerelation} with initial values
\begin{equation}\label{equation:generalinitalconditions}
    P_k(0,y) = {\begin{bmatrix} y\\k\end{bmatrix}_b}\prod_{i=0}^{k-1}\left(cb^y-b^i\right).
\end{equation}
In particular, these become generalised Krawtchouk Polynomials associated with the skew-symmetric matrices of order $t$ with the particular parameter $b=q^2$ then,
\begin{equation}
    P_k(x,n) = \sum_{j=0}^k (-1)^{k-j} q^{2{k-j \choose 2}} {\begin{bmatrix} n-j \\ n-k \end{bmatrix}}{\begin{bmatrix}n-x\\j\end{bmatrix}} q^{jm},
\end{equation}
and in particular,
\begin{equation}
    P_k(0,n)=
    \begin{bmatrix}n\\k\end{bmatrix}\gamma(m,k).
\end{equation}
\begin{note}
From here $\begin{bmatrix}x\\k\end{bmatrix}$ is as defined in Definition \ref{def:skewqnarygaussiancoeffs}.
\end{note}
These initial values, $P_k(0,n)$, count the number of matrices at distance $k$ from any fixed matrix. Now let $\boldsymbol{P}=\left(p_{xk}\right)$ be the $(n+1)\times (n+1)$ matrix with $p_{xk}=P_k(x,n)$.
The matrix $\boldsymbol{P}$ can be used to relate the weight distributions of any code and it's dual. The following theorem is given in \cite{DelsarteAlternating} in relation to alternating bilinear forms but is proved in general for any association scheme in \cite{delsarte1973algebraic}. Here it is written specifically in relation to codes as subgroups of $\mathscr{A}_{q,t}$. It is analogous to the MacWilliams Identity relating the distance distributions of a code and it's dual \cite{TheoryofError}\cite{SystematicCode}.
\begin{thm}\label{thm:DelsarteMacWilliams}
Let $\mathscr{C}\subseteq\mathscr{A}_{q,t}$ be a code with weight distribution $\boldsymbol{c}=(c_0,c_1,\ldots,c_n)$ and  $\boldsymbol{\mathscr{C}}^\perp$ be it's dual with weight distribution $\boldsymbol{c}'=(c'_0,c'_1,\ldots,c'_n)$. Then,
\begin{equation}\label{equation:delsarte'sMacWills}
    \boldsymbol{c}'=\dfrac{1}{\lvert\mathscr{C}\rvert}\boldsymbol{c}\boldsymbol{P}.
\end{equation}
\end{thm}

\subsection{The Skew-$q$-Krawtchouk Polynomials}
We now consider the following set of polynomials which arise in finding the skew-$q$-transform $\overline{a}\left(\mu,\nu;m\right)$ where $a(X,Y;\lambda)$ is as defined in Definition \ref{skew-q-transform} and $\mu(X,Y;m)$ and $\nu(X,Y;m)$ are as in Section \ref{subsection:usingskewqproduct}.
\begin{defn}
For $t\in\mathbb{Z}^+$, $x,k\in\{0,1,\ldots,n\}$ where $n=\lfloor\frac{t}{2}\rfloor,$ and $m=\frac{t(t-1)}{2n}$ we define the \textbf{\textit{the Skew-$\boldsymbol{q}$-Krawtchouk Polynomial}} as
\begin{equation}
    C_k(x,n) = \sum_{j=0}^k (-1)^j q^{2j(n-x)} q^{j(j-1)}{\left[\begin{matrix} x \\ j\end{matrix}\right]}{\left[\begin{matrix} n-x \\ k-j\end{matrix}\right]}\gamma(m-2j,k-j).
\end{equation}
\end{defn}
\begin{note}
    We note that the value of the skew-$q$-Krawtchouk polynomial $C_k(x,n)$ depends on $m$, which in turn depends on the parity of $t$. However, it behaves in the same way regardless of the parity of $t$, and as such we shall use our shorthand notation and only make the dependence on $n$ explicit.
\end{note}
We first prove that the $C_k(x,n)$ satisfy the recurrence relation \eqref{equation:generalisedKpolynomial} and the initial values in \eqref{equation:generalinitalconditions} and are therefore generalised Krawtchouk polynomials.

\begin{prop}
For all $x,k\in\{0,\ldots,n\}$ we have
\begin{equation}
    C_{k+1}(x+1,n+1) = q^{2(k+1)}C_{k+1}(x,n)-q^{2k}C_k(x,n).\label{equation:recurrenceCKI}
\end{equation}

\end{prop}

\begin{proof}

Let $C=C_{k+1}(x+1,n+1)- q^{2(k+1)}C_{k+1}(x,n)+q^{2k}C_k(x,n)$. By definition, 
\begin{align}
    C_{k+1}(x+1,n+1)
    & = \left. C_{k+1}(x+1,n+1)\right|_{j=k+1}\\
    & + \sum_{j=0}^{k} (-1)^j q^{2j(n-x)} q^{j(j-1)} {\left[\begin{matrix} x+1 \\ j\end{matrix}\right]}{\left[\begin{matrix} n-x \\ k+1-j \end{matrix}\right]} \gamma\left(m+2-2j, k+1-j\right)\\
     \overset{\eqref{equation:thing1}}&{=} \left. C_{k+1}(x+1,n+1)\right|_{j=k+1}\\
    & + \sum_{j=0}^{k} (-1)^j q^{2j(n-x)+j(j+1)} {\left[\begin{matrix} x \\ j\end{matrix}\right]}{\left[\begin{matrix} n-x \\ k+1-j\end{matrix}\right]}\gamma\left(m+2-2j, k+1-j\right)\\
    & + \sum_{j=1}^{k} (-1)^j q^{2j(n-x)+j(j-1)} {\left[\begin{matrix} x \\ j-1\end{matrix}\right]}{\left[\begin{matrix} n-x \\ k+1-j\end{matrix}\right]}\gamma\left( m+2-2j, k+1-j\right)\\
    \overset{\eqref{equation:gamma2step}}&{=} \left. C_{k+1}(x+1,n+1)\right|_{j=k+1}\\
    & + \sum_{j=0}^k (-1)^j q^{2j(n-x)+j(j-1)+m+2+2(k-j)}{\left[\begin{matrix} x \\ j\end{matrix}\right]}{\left[\begin{matrix} n-x \\ k+1-j\end{matrix}\right]} \gamma\left(m-2j,k-j\right)\label{equation:alpha}\\
    & - \sum_{j=0}^k (-1)^j q^{2j(n-x)+j(j-1)+2k}  {\left[\begin{matrix} x \\ j\end{matrix}\right]}{\left[\begin{matrix} n-x \\ k+1-j\end{matrix}\right]}\gamma(m-2j,k-j)\label{equation:beta} \\
    & + \sum_{j=0}^{k} (-1)^j q^{2j(n-x)+j(j-1)} {\left[\begin{matrix} x \\ j-1\end{matrix}\right]}{\left[\begin{matrix} n-x \\ k+1-j\end{matrix}\right]} \gamma(m+2-2j, k+1-j)\label{equation:lambda}\\
    & = \alpha - \beta + \lambda + \left. C_{k+1}(x+1,n+1) \right|_{j=k+1}
\end{align}

where $\alpha$, $\beta$, $\lambda$ represent summands \eqref{equation:alpha}, \eqref{equation:beta}, \eqref{equation:lambda} respectively and for notation, $|_{j=k+1}$ means ``the term when $j=k+1$''.

Similarly, 
\begin{equation}
    q^{2(k+1)}C_{k+1}(x,n) = \sum_{j=0}^{k+1}(-1)^j q^{2(k+1)}q^{2j(n-x)} q^{j(j-1)}{\left[\begin{matrix} x \\ j\end{matrix}\right]}{\left[\begin{matrix} n-x \\ k+1-j\end{matrix}\right]}\gamma(m-2j,k+1-j).
\end{equation}

But,

\begin{equation}
    q^{2(k+1)} \gamma(m-2j,k+1-j) \overset{\eqref{equation:gamma1step}}{=} \begin{cases}
     q^{2k}\left(q^{m+2-2j}-q^{2(k-j+1)}\right)\gamma(m-2j,k-j) & \text{if}~j<k+1,\\
     q^{2(k+1)} & \text{if}~ j=k+1.
    \end{cases}
\end{equation}

So,
\begin{align}
 q^{2(k+1)}C_{k+1}(x,n) & =  q^{2(k+1)}\left. C_{k+1}(x,n)\right|_{j=k+1}\\
 & + \sum_{j=0}^{k} (-1)^j q^{2j(n-x)+j(j-1)+m+2+2(k-j)} {\left[\begin{matrix} x \\ j\end{matrix}\right]}{\left[\begin{matrix} n-x \\ k+1-j\end{matrix}\right]} \gamma(m-2j,k-j)\label{equation:delta}\\
 & - \sum_{j=0}^k (-1)^j q^{2j(n-x)+j(j-1)+2k+2(k-j+1)} {\left[\begin{matrix} x \\ j\end{matrix}\right]}{\left[\begin{matrix} n-x \\ k+1-j\end{matrix}\right]}\gamma(m-2j, k-j)\label{equation:varepsilon}\\
 & = \alpha + \varepsilon + q^{2(k+1)}\left. C_{k+1}(x,n)\right|_{j=k+1}.
\end{align}
Where $\varepsilon$ represents the summand \eqref{equation:varepsilon}. 
Thirdly,

\begin{equation}
    \begin{split}
        q^{2k}C_k(x,n) & =\sum_{j=0}^k q^{2j(n-x)+j(j-1)+2k}(-1)^j {\left[\begin{matrix} x \\ j\end{matrix}\right]}{\left[\begin{matrix} n-x \\ k-j\end{matrix}\right]} \gamma(m-2j,k-j),\\
        & = \sigma, ~\text{say}.
    \end{split}
\end{equation}
So we have, 
\begin{equation}
    C =  \beta + \lambda - \varepsilon +\sigma + \left. C_{k+1}(x+1,n+1)\right|_{j=k+1}-q^{2(k+1)} \left. C_{k+1}\right|_{j=k+1}
\end{equation}
and 
\begin{align}
    \beta-\varepsilon & = \sum_{j=0}^k q^{2j(n-x)+j(j-1)+2k}(-1)^{j+1} {\left[\begin{matrix} x \\ j\end{matrix}\right]}{\left[\begin{matrix} n-x \\ k+1-j\end{matrix}\right]}\gamma(m-2j,k-j)\left( 1-q^{2(k-j+1)}\right)\\
    \overset{\eqref{equation:thing3}}&{=} \sum_{j=0}^k q^{2j(n-x)+j(j-1)+2k}(-1)^{j+1}\left(1-q^{2(k-j+1)}\right) {\left[\begin{matrix} x \\ j\end{matrix}\right]}\dfrac{q^{2((n-x)-(k-j))}-1}{q^{2(k+1-j)}-1}{\left[\begin{matrix} n-x \\ k-j\end{matrix}\right]} \gamma(m-2j,k-j)\\
    & = \sum_{j=0}^k q^{2j(n-x)+j(j-1)+2k}(-1)^{j+1} {\left[\begin{matrix} x \\ j\end{matrix}\right]}{\left[\begin{matrix} n-x \\ k-j\end{matrix}\right]}\gamma(m-2j,k-j)\\
    & + \sum_{j=0}^k q^{(j+1)(2n-2x+j)}(-1)^j {\left[\begin{matrix} x \\ j\end{matrix}\right]}{\left[\begin{matrix} n-x \\ k-j\end{matrix}\right]}\gamma(m-2j, k-j)\label{equation:tau}\\
    & = -\sigma +\tau,
\end{align}

where $\tau$ represents the summand in \eqref{equation:tau}.

So $\beta-\varepsilon+\sigma=\tau$. Thus,
\begin{equation}
    C  = \lambda + \tau + \left. C_{k+1}(x+1,n+1)\right|_{j=k+1} - q^{2(k+1)}\left. C_{k+1}(x,n)\right|_{j=k+1}.
\end{equation}

Now, 
\begin{align}
C_{k+1}\left.(x+1,n+1)\right|_{j=k+1} - q^{2(k+1)} \left. C_{k+1}(x,n)\right|_{j=k+1}
    & = q^{2(k+1)(n-x)}(-1)^{k+1}q^{(k+1)k}{\left[\begin{matrix} x + 1 \\ k+1\end{matrix}\right]}\\
    & ~\ ~ -q^{2(k+1)}q^{2(k+1)(n-x)}(-1)^{k+1}q^{(k+1)k}{\left[\begin{matrix} x \\ k+1\end{matrix}\right]}\\
    \overset{\eqref{equation:thing1}}&{=} q^{2(k+1)(n-x)+k(k+1)}(-1)^{k+1} {\left[\begin{matrix} x \\ k\end{matrix}\right]}.
\end{align}

So, 

\begin{equation}
    C = \lambda +\tau + q^{2(k+1)(n-x)+k(k+1)}(-1)^{k+1} {\left[\begin{matrix} x \\ k\end{matrix}\right]}.
\end{equation}
Now, 
\begin{equation}
    \left. \tau\right|_{j=k} = q^{2(k+1)(n-x)+k(k+1)}(-1)^{k} {\left[\begin{matrix} x \\ k\end{matrix}\right]}
\end{equation}

Leaving, 
\begin{equation}
    C=\lambda +\tau -\left. \tau\right|_{j=k}.
\end{equation}

Now consider $\lambda$.
\begin{align}
    \lambda & = \sum_{j=1}^k q^{2j(n-x)+j(j-1)}(-1)^j{\left[\begin{matrix} x \\ j-1\end{matrix}\right]}{\left[\begin{matrix} n-x \\ k+1-j\end{matrix}\right]}\gamma(m+2-2j, k+1-j)\\
    & = \sum_{j=0}^{k-1} q^{(j+1)(2n-2x+j)}(-1)^{j+1}{\left[\begin{matrix} x \\ j\end{matrix}\right]}{\left[\begin{matrix} n-x \\ k-j\end{matrix}\right]}\gamma(m-2j, k-j)\\
    & = -(\tau-\tau|_{j=k}).
\end{align}
So $C=\lambda +\tau -\tau|_{j=k}=0$ and so the $C_k(x,n)$ satisfy the recurrence relation \eqref{equation:recurrenceCKI}.
\end{proof}

\begin{lem}\label{lemma:ckiequalspki}
The $C_k(x,n)$ are the generalised Krawtchouk polynomials. In other words,
\begin{equation}\label{equation:ckiequalspki}
    C_k(x,n) = P_k(x,n).
\end{equation}
\end{lem}
\begin{proof}
The $C_k(x,n)$ satisfy the recurrence relation \eqref{equation:recurrenceCKI} and the initial values of the $C_k(x,n)$ are 
\begin{align}
    C_k(0,n) & = \sum_{j=0}^k (-1)^j q^{2jn}q^{j(j-1)}{\begin{bmatrix} 0\\j\end{bmatrix}}{\begin{bmatrix}n\\k-j\end{bmatrix}}\gamma(m-2j, k-j)\\
    & = {\begin{bmatrix}n\\k\end{bmatrix}}\gamma(m,k).
\end{align}
\end{proof}

We note that this explicit form for the generalised Krawtchouk polynomials is distinct from the three forms presented in \cite[(15)]{delsartereccurance}.

\subsection{The MacWilliams Identity for the Skew Rank Metric}\label{section:Macwilliamsidentity}

We now use the Skew-$q$-Krawtchouk polynomials to prove the $q$-analog form of the MacWilliams Identity for alternating bilinear forms over $\mathbb{F}_q$. We note that this form is similar to the $q$-analog of the MacWilliams Identity developed in \cite{gadouleau2008macwilliams} for linear rank metric codes over $\mathbb{F}_{q^m}$ but differs in the parameters of the $q$-transforms and the meaning of the variable $m$.\\

Let the skew rank weight enumerator of $\mathscr{C}$ be,
\begin{equation}
    W_{\mathscr{C}}^{SR}(X,Y)=\sum_{i=0}^n c_i Y^{i} X^{n-i}
\end{equation}
and of it's dual, $\mathscr{C}^{\perp}$ be
\begin{equation}
    W_{\mathscr{C}^\perp}^{SR}(X,Y)=\sum_{i=0}^n c_i' Y^{i} X^{n-i}.
\end{equation}

\begin{thm}[The MacWilliams Identity for the Skew Rank Metric]\label{mainthm1}
Let $\mathscr{C}$ be a linear code with $\mathscr{C}\subseteq \mathscr{A}_{q,t}$.
Then
\begin{equation}
    W_{\mathscr{C}^\perp}^{SR}(X,Y)=\frac{1}{\left| \mathscr{C}\right|}\overline{W}_{\mathscr{C}}^{SR}\left( X +(q^m-1)Y, X-Y\right).
\end{equation}
\end{thm}

\begin{proof}
For $0\leq i\leq n$ we have
\begin{align}
\left(X-Y\right)^{[i]}\ast\left(X+\left(q^m-1\right)Y\right)^{[n-i]} 
\overset{\eqref{qskewproduct}}&{=} \sum_{k=0}^n\left(\sum_{\ell=0}^k q^{2\ell(n-i)}(-1)^\ell q^{\ell(\ell-1)}{\begin{bmatrix}i\\\ell\end{bmatrix}}{\begin{bmatrix}n-i\\k-\ell\end{bmatrix}}\gamma(m-2\ell, k-\ell)\right)Y^{k}X^{n-k}\\
& = \sum_{k=0}^n C_k(i,n)Y^{k}X^{n-k} \\
\overset{\eqref{equation:ckiequalspki}}&{=} \sum_{k=0}^n P_k(i,n)Y^{k}X^{n-k}.
\end{align}
So then we have
\begin{align}
    \dfrac{1}{\left\vert\mathscr{C} 
    \right\vert} \overline{W}^{SR}_\mathscr{C}\left(X+\left(q^m-1\right)Y, X-Y\right)
    & = \dfrac{1}{\left\vert\mathscr{C}\right\vert}\sum_{i=0}^n c_i \sum_{k=0}^n P_k(i,n)Y^kX^{n-k}\\
    & =\sum_{k=0}^n\left(\dfrac{1}{\left\vert\mathscr{C}\right\vert}\sum_{i=0}^n c_i P_k(i,n)\right)Y^{k}X^{n-k}\\
    \overset{\eqref{equation:delsarte'sMacWills}}&{=} \sum_{k=0}^n c_k' Y^{k}X^{n-k} \\
    & = W_{\mathscr{C}^\perp}^{SR}(X,Y).
\end{align}

\end{proof}

In this way we have shown that the MacWilliams identity for a code and it's dual based on alternating bilinear forms over $\mathbb{F}_q$ can be expressed as a $q$-transform of homogeneous polynomials in a form analogous to the original MacWilliams identity for the Hamming metric and the $q$-analog developed by \cite{gadouleau2008macwilliams} for the rank metric.

\section{The Skew-$q$-Derivatives}\label{section:derivatives}
In this section we develop a new skew-$q$-derivative and skew-$q^{-1}$-derivative to help analyse the coefficients of skew rank weight enumerators. This is analogous to the $q$-derivative applied to the rank metric in \cite{gadouleau2008macwilliams} with the parameter $q$ replaced by $q^2$.

\subsection{The Skew-$q$-Derivative}

\begin{defn}
For $q\geq 2$, the \textbf{\textit{skew-$\boldsymbol{q}$-derivative}} at $X\neq 0$ for a real-valued function $f(X)$ is defined as
\begin{equation}
    f^{(1)}\left(X\right)=\dfrac{f\left(q^2X\right)-f\left(X\right)}{(q^2-1)X}. 
\end{equation}
For $\varphi\geq0$ we denote the $\varphi^{th}$ skew-$q$-derivative (with respect to $X$) of $f(X,Y;\lambda)$ as $f^{(\varphi)}(X,Y;\lambda)$. The $0^{th}$ skew-$q$-derivative of $f(X,Y;\lambda)$ is $f(X,Y;\lambda)$. For any real number $a,~X\neq0,$
\begin{equation}
    \left[ f(X)+ag(X)\right]^{(1)} = f^{(1)}(X)+ag^{(1)}(X).
\end{equation}

\end{defn}

\begin{lem}\label{lemma:vthbderivative}~\\
\begin{enumerate}
\item For $0\leq \varphi \leq \ell, \varphi\in\mathbb{Z}^+, \ell\geq1,$ 
\begin{equation}
    \left(X^\ell\right)^{(\varphi)} = \beta(\ell,\varphi)X^{\ell-\varphi}.
\end{equation}
\item The $\varphi^{th}$ skew-$q$-derivative of $f(X,Y;\lambda)=\displaystyle\sum_{i=0}^r f_i(\lambda) Y^{i}X^{r-i}$ is given by
\begin{equation}\label{equation:vthbderivative}
    f^{(\varphi)}\left(X,Y;\lambda\right)=\displaystyle\sum_{i=0}^{r-\varphi}f_i(\lambda) \beta(r-i,\varphi)Y^{i}X^{r-i-\varphi}.
\end{equation}
\item Also,
\begin{align}
    \mu^{[k](\varphi)}(X,Y;\lambda) & = \beta(k,\varphi)\mu^{[k-\varphi]}(X,Y;\lambda)\label{equation:muderiv}\\
    \nu^{[k](\varphi)}(X,Y;\lambda) & = \beta(k,\varphi)\nu^{[k-\varphi]}(X,Y;\lambda).\label{equation:nuderiv}
\end{align}
\end{enumerate}
\end{lem}
\begin{proof}~\\
\begin{enumerate}
\item[(1)] For $\varphi=1$ we have
\begin{equation}
    \left(X^{\ell}\right)^{(1)} = \dfrac{\left(q^2X^{}\right)^\ell-X^{\ell}}{(q^2-1)X} = \dfrac{q^{2\ell}-1}{q^2-1}X^{\ell-1} = {\begin{bmatrix}\ell\\1\end{bmatrix}}X^{\ell-1} = \beta(\ell,\varphi)X^{\ell-1}.
\end{equation}
The rest of the proof follows by induction on $\varphi$ and is omitted. 
\item[(2)] Now consider $f(X,Y;\lambda)=\displaystyle\sum_{i=0}^r f_i (\lambda)Y^{i}X^{r-i}$. We have,
\begin{align}
    f^{(1)}\left(X,Y;\lambda\right) & = \left(\displaystyle\sum_{i=0}^r f_i (\lambda) Y^{i}X^{r-i}\right)^{(1)}\\
    & =\displaystyle\sum_{i=0}^r f_i(\lambda) Y^{i}\left(X^{r-i}\right)^{(1)}\\
    & =\displaystyle\sum_{i=0}^{r-1}f_i(\lambda) \beta(r-i,\varphi)Y^{i}X^{r-i-1}
\end{align}

The rest of the proof follows by induction on $\varphi$ and is omitted. 
\item[(3)] Now consider $\mu^{[k]}=\displaystyle\sum_{u=0}^k \mu_u(\lambda,k)Y^{u}X^{k-u}$ where $\mu_u(\lambda,k) = {\left[\begin{matrix} k \\ u\end{matrix}\right]}\gamma(\lambda,u)$ as in Theorem \ref{bformula}. Then we have
\begin{align}
    \mu^{[k](1)}(X,Y;\lambda) & = \left(\sum_{u=0}^k \mu_u(\lambda,k)Y^{u}X^{k-u}\right)^{(1)}\\
    & = \sum_{u=0}^k \mu_u(\lambda,k)Y^{u}\left(\frac{\left(q^2X\right)^{k-u}-X^{k-u}}{(q^2-1)X}\right)\\
    & = \sum_{u=0}^{k-1}\frac{q^{2(k-u)}-1}{q^2-1}{\begin{bmatrix}k\\u\end{bmatrix}}\gamma(\lambda,u)Y^{u}X^{k-u-1}\\
    \overset{\eqref{equation:thing2}}&{=}  \sum_{u=0}^{k-1}\frac{(q^{2k}-1)\left(q^{2(k-u)}-1\right)}{(q^{2(k-u)}-1)(q^2-1)}{\begin{bmatrix}k-1\\u\end{bmatrix}}\gamma(\lambda,u)Y^{u}X^{k-u-1}\\
    & = \left(\frac{q^{2k}-1}{q^2-1}\right)\mu^{[k-1]}(X,Y;\lambda)\\
    \overset{\eqref{equation:skewqbetafunction}}&{=} \beta(k,1)\mu^{[k-1]}(X,Y;\lambda).
\end{align}

So $\mu^{[k](\varphi)}(X,Y;\lambda)=\beta(k,\varphi)\mu^{[k-\varphi]}(X,Y;\lambda)$ follows by induction on $\varphi$ and is omitted. \\

Now consider $\nu^{[k]}=\displaystyle\sum_{u=0}^k (-1)^u q^{u(u-1)}{\begin{bmatrix}k\\u\end{bmatrix}}Y^{u}X^{k-u}$ as in Theorem \ref{thed'slemma}. Then we have
\begin{align}
    \nu^{[k](1)}(X,Y;\lambda) & = \sum_{u=0}^k (-1)^u q^{u(u-1)}\frac{q^{2(k-u)}-1}{q^2-1}{\begin{bmatrix}k\\u\end{bmatrix}} Y^{u}X^{k-u-1}\\
    & = \sum_{u=0}^{k-1}(-1)^u q^{u(u-1)}\frac{\left(q^{2k}-1\right)\left(q^{2(k-u)}-1\right)}{\left(q^{2(k-u)}-1\right)(q^2-1)}{\begin{bmatrix}k-1\\u\end{bmatrix}}Y^{u}X^{k-1-u}\\
    & = \frac{q^{2k}-1}{q^2-1}\nu^{[k-1]}(X,Y;\lambda)\\
    & = \beta(k,1)\nu^{[k-1]}(X,Y;\lambda).
\end{align}
So $\nu^{[k](\varphi)}(X,Y;\lambda)=\beta(k,\varphi)\nu^{[k-\varphi]}(X,Y;\lambda)$ follows by induction also and is omitted. 

\end{enumerate}
\end{proof}

We now need a few smaller lemmas in order to prove Leibniz rule for the skew-$q$-derivative. 
\begin{lem}\label{lemma:leibnizsmallqderiv}
Firstly let
\begin{align}
    u\left(X,Y;\lambda\right) & = \sum_{i=0}^r u_i(\lambda) Y^{i}X^{r-i}\\
    v\left(X,Y;\lambda\right) & = \sum_{i=0}^s v_i(\lambda) Y^{i}X^{s-i}.
\end{align}
\begin{enumerate}
\item If $u_r(\lambda)=0$ then
\begin{equation}\label{equation:urequals0}
    \frac{1}{X}\left[u\left(X,Y;\lambda\right)\ast v\left(X,Y;\lambda\right)\right] = \frac{u\left(X,Y;\lambda\right)}{X}\ast v\left(X,Y;\lambda\right).
\end{equation}
\item If $v_s(\lambda)=0$ then
\begin{equation}\label{equation:vsequals0}
    \frac{1}{X}\left[u\left(X,Y;\lambda\right)\ast v\left(X,Y;\lambda\right)\right] = u\left(X, q^2Y;\lambda\right)\ast \frac{v\left(X,Y;\lambda\right)}{X}.
\end{equation}
\end{enumerate}
\end{lem}
\begin{proof}~\\
\begin{enumerate}
    \item[(1)] If $u_r(\lambda)=0$,
\begin{equation}
    \frac{u\left(X,Y;\lambda\right)}{X} = \sum_{i=0}^{r-1}u_i(\lambda) Y^{i}X^{r-i-1}.
\end{equation}
Hence
\begin{align}
    \frac{u\left(X,Y;\lambda\right)}{X}\ast v\left(X,Y;\lambda\right) & = \sum_{k=0}^{r+s-1}\left(\sum_{\ell=0}^k q^{2\ell s} u_\ell(\lambda) v_{k-\ell}(\lambda-2\ell)\right) Y^{k}X^{r+s-1-k}\\
    & = \frac{1}{X}\sum_{k=0}^{r+s-1}\left(\sum_{\ell=0}^k q^{2\ell s} u_\ell(\lambda) v_{k-\ell}(\lambda-2\ell)\right) Y^{k}X^{r+s-k}\\
    & ~ + \frac{1}{X}\sum_{\ell=0}^{r+s} q^{2\ell s} u_\ell(\lambda) v_{r+s-\ell}(\lambda-2\ell) Y^{r+s}X^{0}\\
    & = \frac{1}{X}\left(u\left(X,Y;\lambda\right)\ast v\left(X,Y;\lambda\right)\right)
\end{align}
since $v_{r+s-\ell}(\lambda-2\ell)=0$ for $0\leq \ell \leq r-1$ and $u_{\ell}(\lambda)=0$ for $r\leq \ell \leq r+s$ so $\frac{1}{X}\sum_{\ell=0}^{r+s} q^{2\ell s} u_\ell(\lambda) v_{r+s-\ell}(\lambda-2\ell) Y^{r+s}X^{0} = 0$.
\item[(2)] Now if $v_s(\lambda)=0$,
\begin{equation}
    \frac{v\left(X,Y;\lambda\right)}{X} = \sum_{i=0}^{s-1} v_i(\lambda) Y^{i}X^{s-1-i}.
\end{equation}
Then
\begin{align}
    u\left(X,q^2Y;\lambda\right) \ast \frac{v\left(X,Y;\lambda\right)}{X} & = \sum_{k=0}^{r+s-1}\left(\sum_{\ell=0}^k q^{2\ell(s-1)} q^{2\ell}u_\ell(\lambda) v_{k-\ell}(\lambda-2\ell)\right)Y^{k}X^{r+s-1-k}\\
    & = \sum_{k=0}^{r+s-1}\left(\sum_{\ell=0}^k q^{2\ell(s-1)} q^{2\ell}u_\ell(\lambda) v_{k-\ell}(\lambda-2\ell)\right)Y^{k}X^{r+s-1-k}\\
    & + ~\frac{1}{X}\sum_{\ell=0}^{r+s} q^{2\ell s} u_\ell(\lambda) v_{r+s-\ell}(\lambda-2\ell) Y^{r+s}X^{0}\\
    & = \frac{1}{X}\left[u(X,Y;\lambda)\ast v(X,Y;\lambda)\right]
\end{align}
 since $v_{r+s-\ell}(\lambda-2\ell)=0$ for $0\leq \ell \leq r$ and $u_{\ell}(\lambda)=0$ for $r+1\leq \ell \leq r+s$.
\end{enumerate}
\end{proof}

\begin{thm}[Leibniz rule for the skew-$q$-derivative]\label{Liebnizbderiv}
For two homogeneous polynomials in $X$ and $Y$, $f(X,Y;\lambda)$ and $g(X,Y;\lambda)$ with degrees $r$ and $s$ respectively, the $\varphi^{th}$ (for $\varphi\geq0$) skew-$q$-derivative of their skew-$q$-product is given by

\begin{equation}
    \left[ f\left(X,Y;\lambda\right)\ast g\left(X,Y;\lambda\right)\right]^{(\varphi)} = \sum_{\ell=0}^\varphi {\begin{bmatrix}\varphi \\ \ell \end{bmatrix}}q^{2(\varphi-\ell)(r-\ell)}f^{(\ell)}\left(X,Y;\lambda\right)\ast g^{(\varphi-\ell)}\left(X,Y;\lambda\right).
\end{equation}
\end{thm}

\begin{proof}
Firstly let,
\begin{align}
    f\left(X,Y;\lambda\right) & = \sum_{i=0}^r f_i(\lambda) Y^{i}X^{r-i}\\
    u\left(X,Y;\lambda\right) & = \sum_{i=0}^r u_i(\lambda) Y^{i}X^{r-i}\\
    g\left(X,Y;\lambda\right) & = \sum_{i=0}^s g_i(\lambda) Y^{i}X^{s-i}\\
    v\left(X,Y;\lambda\right) & = \sum_{i=0}^s v_i(\lambda) Y^{i}X^{s-i}.
\end{align}

For simplification, we shall write $f(X,Y;\lambda)$ as $f(X,Y)$ and similarly for the $g(X,Y;\lambda)$. Now by differentiation we have
\begin{align}
    \left[ f\left(X,Y\right)\ast g\left(X,Y\right)\right]^{(1)} & = \frac{f\left(q^2X,Y\right)\ast g\left(q^2X,Y\right)-f\left(X,Y\right)\ast g\left(X,Y\right)}{(q^2-1)X}\\
    & = \frac{1}{(q^2-1)X} \bigg\{ f\left(q^2X,Y\right)\ast g\left(q^2X,Y\right)-f\left(q^2X,Y\right)\ast g\left(X,Y\right)\\
    & ~ + f\left(q^2X,Y\right)\ast g\left(X,Y\right) - f\left(X,Y\right)\ast g\left(X,Y\right)\bigg\}\\
    & = \frac{1}{(q^2-1)X}\left\{ f\left(q^2X,Y\right)\ast\left(g\left(q^2X,Y\right)-g\left(X,Y\right)\right)\right\}\\
    & ~ +\frac{1}{(q^2-1)X}\bigg\{\left(f\left(q^2X,Y\right)-f\left(X,Y\right)\right)\ast g\left(X,Y\right)\bigg\}\\
    \overset{\eqref{equation:vsequals0}}&{=} f\left(q^2X,q^2Y\right) \ast \left\{\frac{g\left(q^2X,Y\right)-g\left(X,Y\right)}{(q^2-1)X}\right\} \\
    \overset{\eqref{equation:urequals0}}&{~+}\left\{\frac{f\left(q^2X,Y\right)-f\left(X,Y\right)}{(q^2-1)X}\right\}\ast g\left(X,Y\right) \\
    & = q^{2r} f\left(X,Y\right)\ast g^{(1)}\left(X,Y\right) + f^{(1)}\left(X,Y\right)\ast g\left(X,Y\right).
\end{align}

So the initial case holds. Assume the statement holds true for $\varphi=\overline{\varphi}$, i.e.

\begin{equation}
\left[ f\left(X,Y\right)\ast g\left(X,Y\right)\right]^{(\overline{\varphi})} = \sum_{\ell=0}^{\overline{\varphi}} {\begin{bmatrix}\overline{\varphi} \\ \ell \end{bmatrix}}q^{2(\overline{\varphi}-\ell)(r-\ell)}f^{(\ell)}\left(X,Y\right)\ast g^{(\overline{\varphi}-\ell)}\left(X,Y\right).
\end{equation}

Now considering $\overline{\varphi}+1$ and for simplicity we write $f(X,Y;\lambda),~g(X,Y;\lambda)$ as $f,g$ we have
\pagebreak
\begin{align}
    \left[f\ast g\right]^{(\overline{\varphi}+1)} & = \left[ \sum_{\ell=0}^{\overline{\varphi}}{\begin{bmatrix}\overline{\varphi}\\\ell\end{bmatrix}}q^{2(\overline{\varphi}-\ell)(r-\ell)}f^{(\ell)}\ast g^{(\overline{\varphi}-\ell)}\right]^{(1)}\\
    & = \sum_{\ell=0}^{\overline{\varphi}}{\begin{bmatrix}\overline{\varphi}\\\ell\end{bmatrix}}q^{2(\overline{\varphi}-\ell)(r-\ell)}\left[f^{(\ell)}\ast g^{(\overline{\varphi}-\ell)}\right]^{(1)}\\
    & = \sum_{\ell=0}^{\overline{\varphi}}{\begin{bmatrix}\overline{\varphi}\\\ell\end{bmatrix}}q^{2(\overline{\varphi}-\ell)(r-\ell)}\left( q^{2(r-\ell)}f^{(\ell)}\ast g^{(\overline{\varphi}-\ell+1)}+f^{(\ell+1)}\ast g^{(\overline{\varphi}-\ell)}\right)\\
    & = \sum_{\ell=0}^{\overline{\varphi}}{\begin{bmatrix}\overline{\varphi}\\\ell\end{bmatrix}}q^{2(\overline{\varphi}-\ell+1)(r-\ell)}f^{(\ell)}\ast g^{(\overline{\varphi}-\ell+1)}+ \sum_{\ell=1}^{\overline{\varphi}+1}{\begin{bmatrix}\overline{\varphi}\\\ell-1\end{bmatrix}}q^{2(\overline{\varphi}-\ell+1)(r-\ell+1)}f^{(\ell)}\ast g^{(\overline{\varphi}-\ell+1)}\\
    & = {\begin{bmatrix}\overline{\varphi}\\0\end{bmatrix}}q^{2(\overline{\varphi}+1)r}f\ast g^{(\overline{\varphi}+1)}+ \sum_{\ell=1}^{\overline{\varphi}}{\begin{bmatrix}\overline{\varphi}\\\ell\end{bmatrix}}q^{2(\overline{\varphi}+1-\ell)(r-\ell)}f^{(\ell)}\ast g^{(\overline{\varphi}-\ell+1)}\\
    & ~ + {\begin{bmatrix}\overline{\varphi}\\\overline{\varphi}\end{bmatrix}}q^{2(\overline{\varphi}+1-\overline{\varphi}-1)(r-\overline{\varphi}-1+1)}f^{(\overline{\varphi}+1)}\ast g + \sum_{\ell=1}^{\overline{\varphi}}{\begin{bmatrix}\overline{\varphi}\\\ell-1\end{bmatrix}}q^{2(\overline{\varphi}+1-\ell)(r-\ell+1)}f^{(\ell)}\ast g^{(\overline{\varphi}-\ell+1)}\\
    & = q^{2(\overline{\varphi}+1)r}f\ast g^{(\overline{\varphi}+1)} + f^{(\overline{\varphi}+1)}\ast g+ \sum_{\ell=1}^{\overline{\varphi}} \left( {\begin{bmatrix}\overline{\varphi}\\\ell\end{bmatrix}} + q^{2(\overline{\varphi}-\ell+1)}{\begin{bmatrix}\overline{\varphi}\\\ell-1\end{bmatrix}}\right) q^{2(\overline{\varphi}-\ell+1)(r-\ell)}f^{(\ell)}\ast g^{(\overline{\varphi}-\ell+1)}\\
    \overset{\eqref{equation:thing5}}&{=} \sum_{\ell=1}^{\overline{\varphi}}{\begin{bmatrix}\overline{\varphi}+1\\\ell\end{bmatrix}}q^{2(\overline{\varphi}+1-\ell)(r-\ell)}f^{(\ell)}\ast g^{(\overline{\varphi}+1-\ell)}+ {\begin{bmatrix}\overline{\varphi}+1\\0\end{bmatrix}}q^{2(\overline{\varphi}+1)r}f\ast g^{(\overline{\varphi}+1)}\\
    & ~ +{\begin{bmatrix}\overline{\varphi}+1\\\overline{\varphi}+1\end{bmatrix}} q^{2(\overline{\varphi}+1-\overline{\varphi}-1)}f^{(\overline{\varphi}+1)}\ast g\\
    & = \sum_{\ell=0}^{\overline{\varphi}+1}{\begin{bmatrix}\overline{\varphi}+1\\\ell\end{bmatrix}}q^{2(\overline{\varphi}+1-\ell)(r-\ell)}f^{(\ell)}\ast g^{(\overline{\varphi}+1-\ell)}.
\end{align}
\end{proof}

\subsection{The Skew-$q^{-1}$-Derivative}

\begin{defn}
For $q\geq 2$, the \textbf{\textit{skew-$\boldsymbol{q^{-1}}$-derivative}} at $Y\neq 0$ for a real-valued function $g(Y)$ is defined as
\begin{equation}
    g^{\{1\}}\left(Y\right)=\dfrac{g\left(q^{-2}Y\right)-g\left(Y\right)}{(q^{-2}-1)Y}. 
\end{equation}
For $\varphi\geq0$ we denote the $\varphi^{th}$ skew-$q^{-1}$-derivative (with respect to $Y$) of $g(X,Y;\lambda)$ as $g^{\{\varphi\}}(X,Y;\lambda)$. The $0^{th}$ skew-$q^{-1}$-derivative of $g(X,Y;\lambda)$ is $g(X,Y;\lambda)$. For any real number $a,~Y\neq0,$
\begin{equation}
    \left[ f(Y)+ag(Y)\right]^{\{1\}} = f^{\{1\}}(Y)+ag^{\{1\}}(Y).
\end{equation}

\end{defn}

\begin{lem}\label{lemma:bminusderivlemma}~\\
\begin{enumerate}
\item For $0\leq \varphi \leq \ell,$ 
\begin{equation}
    \left(Y^\ell\right)^{\{\varphi\}} = q^{2(\varphi(1-\ell)+\sigma_{\varphi})}\beta(\ell,\varphi)Y^{\ell-\varphi}.
\end{equation}
\item The $\varphi^{th}$ skew-$q^{-1}$-derivative of $g(X,Y;\lambda)=\displaystyle\sum_{i=0}^s g_i(\lambda) Y^{i}X^{s-i}$ is given by
\begin{equation}\label{equation:generalbdervispoly}
    g^{\{\varphi\}}\left(X,Y;\lambda\right)=\displaystyle\sum_{i=\varphi}^{s}g_i(\lambda) q^{2(\varphi(1-i)+\sigma_{\varphi})} \beta(i,\varphi)Y^{i-\varphi}X^{s-i}.
\end{equation}
\item Also,
\begin{align}
    \mu^{[k]\{\varphi\}}(X,Y;\lambda) & = q^{-2\sigma_{\varphi}}\beta(k,\varphi)\gamma(\lambda,\varphi)\mu^{[k-\varphi]}(X,Y;\lambda-2\varphi)\\
    \nu^{[k]\{\varphi\}}(X,Y;\lambda) & = (-1)^{\varphi}\beta(k,\varphi)\nu^{[k-\varphi]}(X,Y;\lambda).
\end{align}
 
\end{enumerate}
\end{lem}

\begin{proof}~\\
\begin{enumerate}
\item[(1)] For $\varphi=1$ we have
\begin{align}
    \left(Y^{\ell}\right)^{\{1\}} = \dfrac{\left(q^{-2}Y\right)^{\ell}-Y^{\ell}}{(q^{-2}-1)Y} & = \left(\dfrac{q^{-2\ell}-1}{q^{-2}-1}\right)Y^{\ell-1} \\
    & = \frac{q^2q^{-2\ell}\left(1-q^{2\ell}\right)}{1-q^2}Y^{\ell-1}\\
    & = q^{-2\ell+2}\beta(\ell,1)Y^{\ell-1}.
\end{align}
So the initial case holds. Assume the case for $\varphi =\overline{\varphi}$ holds. Then we have
\begin{align}
    \left(Y^{\ell}\right)^{\{\overline{\varphi}+1\}} & = \left(q^{2(\overline{\varphi}(1-\ell)+\sigma_{\overline{\varphi}})}\beta(\ell,\overline{\varphi})Y^{\ell-\overline{\varphi}}\right)^{\{1\}}\\
    & = q^{2(\overline{\varphi}(1-\ell)+\sigma_{\overline{\varphi}})}\beta(\ell,\overline{\varphi})\frac{q^{-2(\ell-\overline{\varphi})}Y^{\ell-\overline{\varphi}}-Y^{\ell-\overline{\varphi}}}{\left(q^{-2}-1\right)Y}\\
    & = q^{2(\overline{\varphi}(1-\ell)+\sigma_{\overline{\varphi}})}\left(\frac{q^{-2(\ell-\overline{\varphi})}-1}{q^{-2}-1}\right)\beta(\ell,\overline{\varphi})Y^{\ell-\overline{\varphi}-1}\\
    \overset{\eqref{equation:skewqbetafunction}}&{=} q^{2\overline{\varphi}(1-\ell)}q^{\overline{\varphi}(\overline{\varphi}-1)}q^{-2(\ell-\overline{\varphi})}q^2\frac{q^{2(\ell-\overline{\varphi})}-1}{q^2-1}\prod_{i=0}^{\overline{\varphi}-1}{\begin{bmatrix}\ell-i\\1\end{bmatrix}} Y^{\ell-\overline{\varphi}-1}\\
    & = q^{2((\overline{\varphi}+1)(1-\ell)+\sigma_{\overline{\varphi}+1})}\beta(\ell,\overline{\varphi}+1)Y^{\ell-\overline{\varphi}+1}.
\end{align}
Thus the statement holds by induction. 
\item[(2)] Now consider $g(X,Y;\lambda)=\displaystyle\sum_{i=0}^s g_i (\lambda) Y^{i}X^{s-i}$. For $\varphi=1$ we have

\begin{equation}
    g^{\{1\}}\left(X,Y;\lambda\right) = \left(\sum_{i=0}^s g_i(\lambda) Y^{i}X^{s-i}\right)^{\{1\}} = \sum_{i=0}^s g_i(\lambda) \left(Y^{i}\right)^{\{1\}}X^{s-i} = \sum_{i=0}^s g_i(\lambda) q^{2(-i+1)}\beta(i,1)Y^{i-1}X^{s-i}.
\end{equation}
As $\beta(i,1)=0$ when $i=0$ we have
\begin{equation}
    g^{\{1\}}\left(X,Y;\lambda\right) = \sum_{i=1}^s g_i(\lambda) q^{2((1-i)+\sigma_{1})}\beta(i,1)Y^{i-1}X^{s-i}.
\end{equation}
So the initial case holds. Now assume the case holds for $\varphi=\overline{\varphi}$ i.e. \\
$g^{\{\overline{\varphi}\}}\left(X,Y;\lambda\right)=\displaystyle\sum_{i=\overline{\varphi}}^s g_i(\lambda) q^{2\overline{\varphi}(1-i)+2\sigma_{\overline{\varphi}}}\beta(i,\overline{\varphi})Y^{(i-\overline{\varphi})}X^{s-i}$. Then we have
\begin{align}
    g^{\{\overline{\varphi}+1\}}\left(X,Y;\lambda\right) 
    & = \left(\sum_{i=\overline{\varphi}}^s g_i(\lambda) q^{2(\overline{\varphi}(1-i)+\sigma_{\overline{\varphi}})}\beta(i,\overline{\varphi})Y^{i-\overline{\varphi}}\right)^{\{1\}}X^{s-i}\\
    & = \sum_{i=\overline{\varphi}}^s g_i(\lambda) q^{2(\overline{\varphi}(1-i)+\sigma_{\overline{\varphi}})}\beta(i,\overline{\varphi})q^{-2(i-\overline{\varphi}-1)}\beta(i-\overline{\varphi},1)Y^{i-\overline{\varphi}-1}X^{s-i}\\
     \overset{\eqref{equation:skewqbetafunction}}&{=} \sum_{i=\overline{\varphi}}^s g_i(\lambda) q^{2(\overline{\varphi}+1)(1-i)+2\sigma_{\overline{\varphi}}}\prod_{j=0}^{\overline{\varphi}-1}\frac{\left(q^{2(i-j)}-1\right)\left(q^{2(i-\overline{\varphi})}-1\right)}{(q^2-1)(q^2-1)}Y^{i-\overline{\varphi}-1}X^{s-i}\\
    & = \sum_{i=\overline{\varphi}}^s g_i(\lambda) q^{2(\overline{\varphi}+1)(1-i)+2\sigma_{\overline{\varphi}}}\beta(i,\overline{\varphi}+1)Y^{i-\overline{\varphi}-1}X^{s-i}\\
    & = \sum_{i=\overline{\varphi}+1}^s g_i(\lambda) q^{2(\overline{\varphi}+1)(1-i)+2\sigma_{\overline{\varphi}}}\beta(i,\overline{\varphi}+1)Y^{i-\overline{\varphi}-1}X^{s-i}
\end{align}
since when $i=\overline{\varphi}$, $\beta(\overline{\varphi},\overline{\varphi}+1)=0$. So by induction Equation \eqref{equation:generalbdervispoly} holds.
\item[(3)] Now consider $\mu^{[k]}=\displaystyle\sum_{u=0}^k \mu_u(\lambda,k)Y^{u}X^{k-u}$ where $\mu_u(\lambda,k) = {\left[\begin{matrix} k \\ u\end{matrix}\right]}\gamma(\lambda,u)$ as in Theorem \ref{bformula}. Then we have

\begin{align}
    \mu^{[k]\{1\}}(X,Y;\lambda) & = \left(\sum_{u=0}^k \mu_u(\lambda,k)Y^{u}X^{k-u}\right)^{\{1\}}\\
    & = \sum_{u=0}^k \mu_u(\lambda,k)q^{2(1-u)}\beta(u,1)Y^{u-1}X^{k-u}\\
    & = \sum_{r=0}^{k-1} \mu_{r+1}(\lambda,k)q^{2(1-(r+1))}\beta(r+1,1)Y^{r+1-1}X^{k-r-1}\\ 
    \overset{\eqref{equation:thing4}\eqref{equation:gamma2step}}&{=} \sum_{r=0}^{k-1} {\begin{bmatrix}k\\r+1\end{bmatrix}}\gamma(\lambda,r+1)q^{-2r}\beta(r+1,1)Y^{r}X^{k-1-r}\\
    & = \sum_{r=0}^{k-1} \begin{bmatrix}k-1\\r\end{bmatrix} \frac{q^{2k}-1}{q^{2(r+1)}-1}\left(q^\lambda-1\right)q^{2r}q^{-2r}\gamma(\lambda-2,r)\beta(r+1,1)Y^{r}X^{k-1-r}\\
    & = q^{-2\sigma_1}\beta(k,1)\gamma(\lambda,1)\mu^{[k-1]}(X,Y;\lambda-2).
\end{align}

Now assume that the statement holds for $\varphi=\overline{\varphi}$. Then we have
\begin{align}
    \mu^{[k]\{\overline{\varphi}+1\}}(X,Y;\lambda) & = \bigg[ q^{-2\sigma_{\overline{\varphi}}}\beta(k,\overline{\varphi})\gamma(\lambda,\overline{\varphi})\mu^{[k-\overline{\varphi}]}(X,Y;\lambda-2\overline{\varphi})\bigg]^{\{1\}}\\
    & = q^{-2\sigma_{\overline{\varphi}}}\beta(k,\overline{\varphi})\gamma(\lambda,\overline{\varphi})\big[\mu^{[k-\overline{\varphi}]}(X,Y;\lambda-2\overline{\varphi})\big]^{\{1\}}\\
    & = q^{-2\sigma_{\overline{\varphi}}}\beta(k,\overline{\varphi})\gamma(\lambda,\overline{\varphi})\left(\sum_{r=0}^{k-\overline{\varphi}}{\begin{bmatrix}k-\overline{\varphi}\\r\end{bmatrix}}\gamma(\lambda-2\overline{\varphi},r)Y^{r}X^{k-\overline{\varphi}-r}\right)^{\{1\}}\\
    & = q^{-2\sigma_{\overline{\varphi}}}\beta(k,\overline{\varphi})\gamma(\lambda,\overline{\varphi})\sum_{r=1}^{k-\overline{\varphi}}{\begin{bmatrix}k-\overline{\varphi}\\r\end{bmatrix}}\gamma(\lambda-2\overline{\varphi},r)\left(Y^{r}\right)^{\{1\}}X^{k-\overline{\varphi}-r}\\
    & = q^{-2\sigma_{\overline{\varphi}}}\beta(k,\overline{\varphi})\gamma(\lambda,\overline{\varphi})\sum_{u=0}^{k-\overline{\varphi}-1}{\begin{bmatrix}k-\overline{\varphi}\\u+1\end{bmatrix}}\gamma(\lambda-2\overline{\varphi},u+1)q^{2(1-(u+1))}\beta(u+1,1)Y^{u+1-1}X^{k-\overline{\varphi}-u-1}\\
    \overset{\eqref{equation:thing4}\eqref{equation:gamma2step}}&{=} q^{-2\sigma_{\overline{\varphi}}}\beta(k,\overline{\varphi})\gamma(\lambda,\overline{\varphi})\sum_{u=0}^{k-(\overline{\varphi}+1)}{\begin{bmatrix}k-\overline{\varphi}-1\\u\end{bmatrix}}\frac{\left(q^{2(k-\overline{\varphi})}-1\right)\left(q^{2(u+1)}-1\right)}{\left(q^{2(u+1)}-1\right)(q^2-1)}q^{2u}q^{-2u}\\
    & ~ \times \left(q^{\lambda-2\overline{\varphi}}-1\right)\gamma(\lambda-2(\overline{\varphi}+1),u)Y^{u}X^{(k-(\overline{\varphi}+1)-u)}\\
    & = q^{-2\sigma_{\overline{\varphi}}}q^{-2\overline{\varphi}}\gamma(\lambda,\overline{\varphi}+1)\beta(k,\overline{\varphi}+1)\mu^{[k-(\overline{\varphi}+1)]}(X,Y;\lambda-2(\overline{\varphi}+1))\\
    & = q^{-2\sigma_{(\overline{\varphi}+1)}}\gamma(\lambda,\overline{\varphi}+1)\beta(k,\overline{\varphi}+1)\mu^{[k-(\overline{\varphi}+1)]}(X,Y;\lambda-2(\overline{\varphi}+1)).
\end{align}
As required.
Now consider $\nu^{[k]}=\displaystyle\sum_{u=0}^k (-1)^u q^{u(u-1)}{\begin{bmatrix}k\\u\end{bmatrix}}Y^{u}X^{k-u}$ as defined in Theorem \ref{thed'slemma}. Then we have
\begin{align}
    \nu^{[k]\{1\}}(X,Y;\lambda) & = \left(\sum_{u=0}^k (-1)^u q^{u(u-1)}{\begin{bmatrix}k\\u\end{bmatrix}}Y^{u}X^{k-u}\right)^{\{1\}}\\
    & = \sum_{u=1}^k (-1)^u q^{u(u-1)}{\begin{bmatrix}k\\u\end{bmatrix}}\left(Y^{u}\right)^{\{1\}}X^{k-u}\\
    & = \sum_{r=0}^{k-1} (-1)^{(r+1)} q^{r(r+1)}q^{2(1-(r+1))}{\begin{bmatrix}k\\r+1\end{bmatrix}}\beta(r+1,1)Y^{r+1-1}X^{k-r-1}\\
    \overset{\eqref{equation:thing4}\eqref{equation:gamma2step}}&{=} -\sum_{r=0}^{k-1} (-1)^{r} q^{r(r-1)}q^{2r}q^{-2r}{\begin{bmatrix}k-1\\r\end{bmatrix}}\frac{\left(q^{2k}-1\right)\left(q^{2(r+1)}-1\right)}{\left(q^{2(r+1)}-1\right)\left(q^2-1\right)}\beta(r,1)Y^{r}X^{k-r-1}\\
    & = (-1)^{1}\beta(k,1)\nu^{[k-1]}(X,Y;\lambda).
\end{align}
Now assume that the statement holds for $\varphi=\overline{\varphi}$. Then we have
\begin{align}
    \nu^{[k]}(X,Y;\lambda)^{\{\overline{\varphi}+1\}} & = \left[(-1)^{\overline{\varphi}}\beta(k,\overline{\varphi})\nu^{[k-\overline{\varphi}]}(X,Y;\lambda)\right]^{\{1\}}\\
    & = (-1)^{\overline{\varphi}}\beta(k,\overline{\varphi})\sum_{u=1}^{k-\overline{\varphi}} (-1)^u q^{u(u-1)}{\begin{bmatrix}k-\overline{\varphi}\\u\end{bmatrix}}\left(Y^{u}\right)^{\{1\}}X^{k-\overline{\varphi}-u}\\
    & = (-1)^{\overline{\varphi}}\beta(k,\overline{\varphi})\sum_{r=0}^{k-\overline{\varphi}-1} (-1)^{r+1} q^{r(r+1)}q^{-2(r+1)+2}{\begin{bmatrix}k-\overline{\varphi}\\r+1\end{bmatrix}}\beta(r+1,1)Y^{r+1-1}X^{k-\overline{\varphi}-r-1}\\
    & = (-1)^{\overline{\varphi}+1}\beta(k,\overline{\varphi})\sum_{r=0}^{k-\overline{\varphi}-1} (-1)^{r}q^{r(r-1)}{\begin{bmatrix}k-(\overline{\varphi}+1)\\r\end{bmatrix}}\\
    & ~ \times \frac{\left(q^{2(k-\overline{\varphi})}-1\right)\left(q^{2(r+1)}-1\right)}{\left(q^{2(r+1)}-1\right)\left(q^2-1\right)}Y^{r}X^{k-\overline{\varphi}-1-r}\\
    & = (-1)^{\overline{\varphi}+1}\beta(k,\overline{\varphi}+1)\nu^{[k-(\overline{\varphi}+1)]}(X,Y;\lambda).
\end{align}
as required.
\end{enumerate}
\end{proof}

Now we need a few smaller lemmas in order to prove Leibniz rule for the skew-$q^{-1}$-derivative.
\begin{lem}\label{lemma:skewqminussimplification} Firstly let
\begin{align}
    u\left(X,Y;\lambda\right) & = \sum_{i=0}^r u_i(\lambda) Y^{i}X^{r-i}\\
     v\left(X,Y;\lambda\right) & = \sum_{i=0}^s v_i(\lambda) Y^{i}X^{s-i}.
\end{align}
\begin{enumerate}
\item If $u_0=0$ then
\begin{equation}
    \frac{1}{Y}\left[u\left(X,Y;\lambda\right)\ast v\left(X,Y;\lambda\right)\right] = q^{2s}\frac{u\left(X,Y;\lambda\right)}{Y}\ast v\left(X,Y;\lambda-2\right).
\end{equation}
\item If $v_0=0$ then
\begin{equation}
    \frac{1}{Y}\left[u\left(X,Y;\lambda\right)\ast v\left(X,Y;\lambda\right)\right] = u\left(X, q^2Y;\lambda\right)\ast \frac{v\left(X,Y;\lambda\right)}{Y}.
\end{equation}
\end{enumerate}
\end{lem}
\begin{proof}~\\
\begin{enumerate}
 \item[(1)] Suppose $u_0=0$. Then
\begin{equation}
    \frac{u\left(X,Y;\lambda\right)}{Y} = \sum_{i=0}^{r}u_i(\lambda) Y^{i-1}X^{r-i} = \sum_{i=0}^{r-1} u_{i+1}(\lambda)Y^{i}X^{r-i-1}
\end{equation}
Hence
\begin{align}
    q^{2s}\frac{u\left(X,Y;\lambda\right)}{Y}\ast v\left(X,Y;\lambda-2\right) & = q^{2s}\sum^{r+s-1}_{u=0}\left(\sum_{\ell=0}^u q^{2\ell s}u_{\ell+1}(\lambda)v_{u-\ell}(\lambda-2u-2)\right) Y^{u}X^{r+s-1-u}\\
    & = q^{2s}\sum_{u=0}^{r+s-1}\left(\sum_{i=1}^{u+1}q^{2(i-1)s}u_i(\lambda)v_{u-i+1}(\lambda-2u-2)\right)Y^{u}X^{r+s-1-u}\\
    & = q^{2s}\sum_{j=1}^{r+s}\left(\sum_{i=1}^{j}q^{2(i-1)s}u_i(\lambda)v_{j-i}(\lambda-2j)\right)Y^{j-1}X^{r+s-j}\\
    & = \frac{1}{Y}\sum_{j=0}^{r+s}\left(\sum_{i=0}^{j}q^{2is}u_i(\lambda)v_{j-i}(\lambda-2j)\right)Y^{j}X^{r+s-j}\\
    & = \frac{1}{Y}\left(u\left(X,Y;\lambda\right)\ast v\left(X,Y;\lambda\right)\right).
\end{align}
 \item[(2)] Now if $v_0=0$, then
\begin{align}
    \frac{v\left(X,Y;\lambda\right)}{Y} & = \sum_{j=1}^s v_j(\lambda)Y^{j-1}X^{s-j}\\
    & = \sum_{i=0}^{s-1} v_{i+1}(\lambda)Y^{i}X^{s-i-1}.
\end{align}
So, 
\begin{align}
    u\left(X,q^2Y;\lambda\right) \ast \frac{v\left(X,Y;\lambda\right)}{Y} & = \sum_{u=0}^{r+s-1}\left(\sum_{j=0}^{u} q^{2j(s-1)}q^{2j}u_j(\lambda)v_{u-j+1}(\lambda-2j)\right)Y^{u}X^{r+s-1-u}\\
    & = \sum_{\ell=1}^{r+s}\left(\sum_{j=0}^{\ell-1} q^{2js}u_j(\lambda)v_{\ell-j}(\lambda-2j)\right)Y^{\ell-1}X^{r+s-\ell}\\
    & = \frac{1}{Y}\sum_{\ell=1}^{r+s}\left(\sum_{j=0}^{\ell} q^{2js}u_j(\lambda)v_{\ell-j}(\lambda-2j)\right)Y^{\ell}X^{r+s-\ell}\\
    & = \frac{1}{Y}\sum_{\ell=0}^{r+s}\left(\sum_{j=0}^{\ell} q^{2js}u_j(\lambda)v_{\ell-j}(\lambda-2j)\right)Y^{\ell}X^{r+s-\ell}\\
    & = \frac{1}{Y}\left(u\left(X,Y;\lambda\right)\ast v\left(X,Y;\lambda\right)\right).
\end{align}
\end{enumerate}
\end{proof}

\begin{thm}[Leibniz rule for the skew-$q^{-1}$-derivative]\label{Liebnizbminusderiv}
For two homogeneous polynomials in $Y$, $f(X,Y;\lambda)$ and $g(X,Y;\lambda)$ with degrees $r$ and $s$ respectively, the $\varphi^{th}$ (for $\varphi\geq0$) skew-$q^{-1}$-derivative of their skew-$q$-product is given by

\begin{equation}
    \left[ f\left(X,Y;\lambda\right)\ast g\left(X,Y;\lambda\right)\right]^{\{\varphi\}} = \sum_{\ell=0}^\varphi {\begin{bmatrix}\varphi \\ \ell \end{bmatrix}}q^{2\ell(s-\varphi+\ell)}f^{\{\ell\}}\left(X,Y;\lambda\right)\ast g^{\{\varphi-\ell\}}\left(X,Y;\lambda-2\ell\right).
\end{equation}
\end{thm}

\begin{proof}
Firstly let,
\begin{align}
    f\left(X,Y;\lambda\right) & = \sum_{i=0}^r f_i(\lambda) Y^{i}X^{r-i}\\
    u\left(X,Y;\lambda\right) & = \sum_{i=0}^r u_i(\lambda) Y^{i}X^{r-i}\\
    g\left(X,Y;\lambda\right) & = \sum_{i=0}^s g_i(\lambda) Y^{i}X^{s-i}\\
    v\left(X,Y;\lambda\right) & = \sum_{i=0}^s v_i(\lambda) Y^{i}X^{s-i}.\\
\end{align}

For simplification we shall write $f(X,Y;\lambda),~g(X,Y;\lambda)$ as $f(Y;\lambda),~g(Y;\lambda)$. Now by differentiation we have
\begin{align}
    \left[ f\left(Y;\lambda\right)\ast g\left(Y;\lambda\right)\right]^{\{1\}} & = \frac{f\left(q^{-2}Y;\lambda\right)\ast g\left(q^{-2}Y;\lambda\right)-f\left(Y;\lambda\right)\ast g\left(Y;\lambda\right)}{(q^{-2}-1)Y}\\
    & = \frac{1}{(q^{-2}-1)Y} \bigg\{ f\left(q^{-2}Y;\lambda\right)\ast g\left(q^{-2}Y;\lambda\right)-f\left(q^{-2}Y;\lambda\right)\ast g\left(Y;\lambda\right)\\
    & ~ + f\left(q^{-2}Y;\lambda\right)\ast g\left(Y;\lambda\right) - f\left(Y;\lambda\right)\ast g\left(Y;\lambda\right)\bigg\}\\
    & = \frac{1}{(q^{-2}-1)Y}\bigg\{ f\left(q^{-2}Y;\lambda\right)\ast\left(g\left(q^{-2}Y;\lambda\right)-g\left(Y;\lambda\right)\right)\bigg\}\\
    & ~ +\frac{1}{(q^{-2}-1)Y}\bigg\{\left(f\left(q^{-2}Y;\lambda\right)-f\left(Y;\lambda\right)\right)\ast g\left(Y;\lambda\right)\bigg\}.
\end{align}
By Lemma \ref{lemma:skewqminussimplification} we have
\begin{align}
    \left[ f\left(Y;\lambda\right)\ast g\left(Y;\lambda\right)\right]^{\{1\}} & = f\left(Y;\lambda\right)\ast\frac{\left(g\left(q^{-2}Y;\lambda\right)-g\left(Y;\lambda\right)\right)}{\left(q^{-2}-1\right)Y}\\
    &~+q^{2s}\frac{\left(f\left(q^{-2}Y;\lambda\right)-f\left(Y;\lambda\right)\right)}{\left(q^{-2}-1\right)Y}\ast g\left(Y;\lambda-2\right)\\
    & = f\left(Y;\lambda\right)\ast g^{\{1\}}\left(Y;\lambda\right) + q^{2s} f^{\{1\}}\left(Y;\lambda\right)\ast g\left(Y;\lambda-2\right).
\end{align}
So the initial case holds. Assume the statement holds true for $\varphi=\overline{\varphi}$, i.e.

\begin{equation}
    \left[ f\left(X,Y;\lambda\right)\ast g\left(X,Y;\lambda\right)\right]^{\{\overline{\varphi}\}} = \sum_{\ell=0}^{\overline{\varphi}} {\begin{bmatrix}\overline{\varphi} \\ \ell \end{bmatrix}}q^{2\ell(s-\overline{\varphi}+\ell)}f^{\{\ell\}}\left(X,Y;\lambda\right)\ast g^{\{\overline{\varphi}-\ell\}}\left(X,Y;\lambda-2r\right).
\end{equation}
Now considering $\overline{\varphi}+1$  and for simplicity we write $f(X,Y;\lambda),~g(X,Y;\lambda)$ as $f(\lambda),g(\lambda)$ we have

\begin{align}
    \left[ f\left(\lambda\right)\ast g\left(\lambda\right)\right]^{\{\overline{\varphi}+1\}} & = \left[\sum_{\ell=0}^{\overline{\varphi}} {\begin{bmatrix}\overline{\varphi} \\ \ell \end{bmatrix}}q^{2\ell(s-\overline{\varphi}+\ell)}f^{\{\ell\}}\left(\lambda\right)\ast g^{\{\overline{\varphi}-\ell\}}\left(\lambda-2\ell\right)\right]^{\{1\}}\\
    & = \sum_{\ell=0}^{\overline{\varphi}} {\begin{bmatrix}\overline{\varphi} \\ \ell \end{bmatrix}}q^{2l(s-\overline{\varphi}+\ell)}\left(f^{\{\ell\}}\left(\lambda\right)\ast g^{\{\overline{\varphi}-\ell\}}\left(\lambda-2\ell\right)\right)^{\{1\}}\\
    \overset{\eqref{qskewproduct}}&{=} \sum_{\ell=0}^{\overline{\varphi}} {\begin{bmatrix}\overline{\varphi} \\ \ell \end{bmatrix}}q^{2\ell(s-\overline{\varphi}+\ell)}f^{\{\ell\}}\left(\lambda\right)\ast g^{\{\overline{\varphi}-\ell+1\}}\left(\lambda-2\ell\right)\\
    & ~+ \sum_{\ell=0}^{\overline{\varphi}} {\begin{bmatrix}\overline{\varphi} \\ \ell \end{bmatrix}}q^{2\ell(s-\overline{\varphi}+\ell)}q^{2(v-\overline{\varphi}+\ell)}f^{\{\ell+1\}}\left(\lambda\right)\ast g^{\{\overline{\varphi}-\ell\}}\left(\lambda-2\ell-2\right)\\
    & = \sum_{\ell=0}^{\overline{\varphi}} {\begin{bmatrix}\overline{\varphi} \\ \ell \end{bmatrix}}q^{2\ell(s-\overline{\varphi}+\ell)}f^{\{\ell\}}\left(\lambda\right)\ast g^{\{\overline{\varphi}-\ell+1\}}\left(\lambda-2\ell\right)\\
    & ~+ \sum_{\ell=1}^{\overline{\varphi}+1} {\begin{bmatrix}\overline{\varphi} \\ \ell-1 \end{bmatrix}}q^{2(\ell-1)(s-\overline{\varphi}+\ell-1)}q^{2(s-\overline{\varphi}+(\ell-1))}f^{\{\ell\}}\left(\lambda\right)\ast g^{\{\overline{\varphi}-k+1\}}\left(\lambda-2\ell\right)\\
    & = f\left(\lambda\right)\ast g^{\{\overline{\varphi}+1\}}\left(\lambda\right)+\sum_{\ell=1}^{\overline{\varphi}} {\begin{bmatrix}\overline{\varphi} \\ \ell \end{bmatrix}}q^{2\ell(s-\overline{\varphi}+\ell)}f^{\{\ell\}}\left(\lambda\right)\ast g^{\{\overline{\varphi}-\ell+1\}}\left(\lambda-2\ell\right)\\
    & ~+ {\begin{bmatrix}\overline{\varphi}\\\overline{\varphi}\end{bmatrix}}q^{2(\overline{\varphi}+1)(s+1)}q^{-2\overline{\varphi}-2} f^{\{\overline{\varphi}+1\}}\left(\lambda\right)\ast g\left(\lambda-2(\overline{\varphi}+1)\right)\\
    & ~+ \sum_{\ell=1}^{\overline{\varphi}} {\begin{bmatrix}\overline{\varphi} \\ \ell-1 \end{bmatrix}}q^{2(\ell-1)(s-\overline{\varphi}+\ell-1)}q^{2(s-\overline{\varphi}+(\ell-1))}f^{\{\ell\}}\left(\lambda\right)\ast g^{\{\overline{\varphi}-k+1\}}\left(\lambda-2\ell\right)\\
    & = f\left(\lambda\right)\ast g^{\{\overline{\varphi}+1\}}\left(\lambda\right)+ \sum_{\ell=1}^{\overline{\varphi}}\left({\begin{bmatrix}\overline{\varphi}\\\ell\end{bmatrix}}+q^{-2\ell}{\begin{bmatrix}\overline{\varphi}\\\ell-1\end{bmatrix}}\right)q^{2\ell(s-\overline{\varphi}+\ell)}f^{\{\ell\}}\left(\lambda\right) \ast g^{\{\overline{\varphi}+1-\ell\}}\left(\lambda-2\ell\right)\\
    & ~ + q^{2s(\overline{\varphi}+1)}q^{-2\overline{\varphi}-2}f^{\{\overline{\varphi}+1\}}\left(\lambda\right) \ast g\left(\lambda-2(\overline{\varphi}+1)\right)\\
    \overset{\eqref{equation:thing1}}&{=} f\left(\lambda\right)\ast g^{\{\overline{\varphi}+1\}}\left(\lambda\right) + \sum_{\ell=1}^{\overline{\varphi}}q^{-2\ell}{\begin{bmatrix}\overline{\varphi}+1\\\ell\end{bmatrix}}f^{\{\ell\}}\left(\lambda\right)\ast g^{\{\overline{\varphi}+1\}}\left(\lambda-2\ell\right)\\
    &~+ q^{-2(\overline{\varphi}+1)}{\begin{bmatrix}\overline{\varphi}+1\\\overline{\varphi}+1\end{bmatrix}}q^{2(\overline{\varphi}+1)(s-\overline{\varphi}(\overline{\varphi}+1))}f^{\{\overline{\varphi}+1\}}\left(\lambda\right)\ast g^{\{\overline{\varphi}+1-(\overline{\varphi}+1)\}}\left(\lambda-2(\overline{\varphi}+1)\right)\\
    & = \sum_{\ell=0}^{\overline{\varphi}+1}{\begin{bmatrix}\overline{\varphi}+1\\\ell\end{bmatrix}}q^{2\ell(s-(\overline{\varphi}+1)+\ell)}f^{\{\ell\}}\left(\lambda\right) \ast g^{\{\overline{\varphi}+1-\ell\}}\left(\lambda-2\ell\right)
\end{align}
as required.
\end{proof}

\subsection{Evaluating the Skew-$q$-Derivative and the Skew-$q^{-1}$-Derivative}

The following lemmas yield useful results for applying the MacWilliams Identity to develop moments of the skew rank distribution.

\begin{lem}\label{lemma:d'sandbetab's}
For $X=Y=1$,
\begin{equation}\label{equation:d'sandbetab's}
    \nu^{[j](\ell)}(1,1;\lambda)=\beta(j,j)\delta_{j\ell}.
\end{equation}
\end{lem}
\begin{proof}
Consider 
\begin{align}
    \nu^{[j](\ell)}(X,Y;\lambda) \overset{\eqref{equation:nuderiv}}&{=}  \beta(j,\ell)\nu^{[j-\ell]}(X,Y;\lambda)\\
    & = \beta(j,\ell)\sum_{u=0}^{j-\ell}(-1)^uq^{u(u-1)}{\begin{bmatrix}j-\ell\\u\end{bmatrix}}Y^{u}X^{(j-\ell)-u}.
\end{align}
So \begin{equation}
    \nu^{[j](\ell)}(1,1;\lambda) = \beta(j,\ell)\sum_{u=0}^{j-\ell}(-1)^uq^{u(u-1)}{\begin{bmatrix}j-\ell\\u\end{bmatrix}}.
\end{equation}
Now the rest of the proof follows directly from Equation \eqref{equation:deltaijbs}. 
\end{proof}

\begin{lem}\label{lemma:rhoandmu}
For any homogeneous polynomial, $\rho\left(X,Y;\lambda\right)$ and for any $s\geq 0$, \begin{equation}\label{equation:rhoandmu}\left(\rho \ast \mu^{[s]}\right)\left(1,1;\lambda\right) =q^{\lambda s}\rho(1,1;\lambda).
\end{equation}
\end{lem}
\begin{proof}
Let $\rho\left(X,Y;\lambda\right)=\displaystyle\sum_{i=0}^r \rho_i(\lambda)Y^{i}X^{r-i}$, then from Theorem \ref{bformula},
\begin{equation}
    \mu^{[s]}(X,Y;\lambda) = \sum_{t=0}^{s}{\begin{bmatrix}s\\t\end{bmatrix}}\gamma(\lambda,t)Y^{t}X^{s-t} = \sum_{t=0}^s \mu^{[s]}_t(\lambda)Y^{t}X^{s-t}
\end{equation}
and
\begin{equation}
    \left(\rho\ast \mu^{[s]}\right)(X,Y;\lambda) = \sum_{u=0}^{r+s}c_u(\lambda)Y^{u}X^{(r+s-u)}
\end{equation}
where
\begin{equation}
    c_u(\lambda)=\sum_{i=0}^u q^{2is}\rho_i(\lambda)\mu^{[s]}_{u-i}(\lambda-2i).
\end{equation}
Then
\begin{align}
    \left(\rho\ast \mu^{[s]}\right)(1,1;\lambda) & = \sum_{u=0}^{r+s} c_u(\lambda)\\
    & = \sum_{u=0}^{r+s}\sum_{i=0}^{u} q^{2is}\rho_i(\lambda)\mu^{[s]}_{u-i}(\lambda-2i)\\
    & = \sum_{j=0}^{r+s}q^{2js}\rho_j(\lambda)\left(\sum_{k=0}^{r+s-j}\mu_k^{[s]}(\lambda-2j)\right)\\
    & = \sum_{j=0}^{r}q^{2js}\rho_j(\lambda)\left(\sum_{k=0}^s \mu_k^{[s]}(\lambda-2j)\right)\\
    & = \sum_{j=0}^r q^{2js} \rho_j(\lambda)\left(\sum_{k=0}^s {\begin{bmatrix}s\\k\end{bmatrix}}\gamma(\lambda-2j,k)\right)\\
    \overset{\eqref{equation:producttosumgauss}}&{=} \sum_{j=0}^r q^{2js} \rho_j(\lambda)q^{(\lambda-2j)s}\\
    & = q^{\lambda s}\rho(1,1;\lambda).
\end{align}
\end{proof}
\section{Moments of the Skew Rank Distribution}\label{section:moments}

Here we explore the moments of the skew rank distribution of a subgroup of alternating bilinear forms over $\mathbb{F}_q$ and that of it's dual. Similar results for the Hamming metric were derived in \cite[p131]{TheoryofError} and for rank metric codes over $\mathbb{F}_{q^m}$ in \cite[Prop 4]{gadouleau2008macwilliams}. 

\subsection{Moments derived from the Skew-$q$-Derivative}

\begin{prop}\label{prop:momentsbderiv}

For $0 \le \varphi \le n$ and a linear code $\mathscr{C} \subseteq \mathscr{A}_{q,t}$ and its dual $\mathscr{C}^\perp$ with weight distributions ${\boldsymbol{c}}$ and ${\boldsymbol{ c'}}$, respectively we have
\begin{equation}
    \sum_{i=0}^{n-\varphi}{\begin{bmatrix} n-i \\ \varphi\end{bmatrix}}c_i = \frac{1}{|\mathscr{C}^\perp|}q^{m(n-\varphi)}\sum_{i=0}^{\varphi} {\begin{bmatrix}n-i\\n-\varphi\end{bmatrix}}c_i^{'}.
\end{equation}

\end{prop}

\begin{proof}
We apply Theorem \ref{mainthm1} to $\mathscr{C}^\perp$ to get
\begin{equation}
    W_{\mathscr{C}}^{SR}(X,Y) = \frac{1}{\left\vert \mathscr{C}^{\perp}\right|} \overline{W}_{\mathscr{C}^\perp}^{SR}\left(X+\left(q^m-1\right)Y,X-Y\right)
\end{equation}
or equivalently
\begin{align}
    \sum_{i=0}^n c_i Y^{i}X^{n-i}
        & =\frac{1}{\left\vert \mathscr{C}^{\perp}\right|}\sum_{i=0}^n c_i'\left(X-Y\right)^{[i]}\ast \left[X+\left(q^m-1\right)Y\right]^{[n-i]}\\
        & = \frac{1}{\left\vert \mathscr{C}^{\perp}\right|} \sum_{i=0}^n c_i'\nu^{[i]}(X,Y;m)\ast \mu^{[n-i]}(X,Y;m). \label{equation:skew-q}
\end{align}

For each side of Equation \eqref{equation:skew-q}, we shall apply the skew-$q$-derivative $\varphi$ times and then evaluate at $X=Y=1$.

For the left hand side, we obtain
\begin{equation}
    \left(\sum_{i=0}^n c_i Y^{i}X^{n-i}\right)^{(\varphi)}=
        \sum_{i=0}^{n-\varphi}c_i \beta(n-i,\varphi)Y^{i}X^{(n-i-\varphi)}
\end{equation}
from Equation \eqref{equation:vthbderivative}. Putting $X=Y=1$ we then get
\begin{align}
    \sum_{i=0}^{n-\varphi} c_i\beta(n-i,\varphi) & = \sum_{i=0}^{n-\varphi}c_i{\begin{bmatrix}n-i\\\varphi\end{bmatrix}}\beta(\varphi,\varphi)\\ 
    & = \beta(\varphi,\varphi)\sum_{i=0}^{n-\varphi}c_i {\begin{bmatrix}n-i\\\varphi\end{bmatrix}}.
\end{align}

We now move on to the right hand side. For simplicity we write $\mu(X,Y;m)$ as $\mu$ and similarly for $\nu(X,Y;m)$. We also get by Theorem \ref{Liebnizbderiv},
\begin{align}
    \left(\frac{1}{\left\vert \mathscr{C}^\perp\right|}\sum_{i=0}^n c_i' \nu^{[i]}\ast \mu^{[n-i]}\right)^{(\varphi)} & = \frac{1}{\left\vert \mathscr{C}^\perp\right|}\sum_{i=0}^n c_i'\left(\sum_{\ell=0}^{\varphi}{\begin{bmatrix}\varphi\\\ell\end{bmatrix}}q^{2(\varphi-\ell)(i-\ell)}\nu^{[i](\ell)}\ast \mu^{[n-i](\varphi-\ell)}\right)\\
    & = \frac{1}{\left\vert \mathscr{C}^\perp\right|}\sum_{i=0}^n c_i'\psi_i.
\end{align}
Then with $X=Y=1$,
\begin{align}
    \psi_i(X,Y;m) & = \sum_{\ell=0}^{\varphi}{\begin{bmatrix}\varphi\\\ell\end{bmatrix}}q^{2(\varphi-\ell)(i-\ell)}\nu^{[i](\ell)}(X,Y;m) \ast \mu^{[n-i](\varphi-\ell)}(X,Y;m)\\
    \psi_i(1,1;m) \overset{\eqref{equation:muderiv}}&{=} \sum_{\ell=0}^{\varphi} {\begin{bmatrix}\varphi\\\ell\end{bmatrix}}q^{2(\varphi-\ell)(i-\ell)}\beta(n-i,\varphi-\ell)\left(\nu^{[i](\ell)}\ast \mu^{[n-i-\varphi+\ell]}\right)(1,1;m)\\
    \overset{\eqref{equation:rhoandmu}}&{=} \sum_{\ell=0}^{\varphi} {\begin{bmatrix}\varphi\\\ell\end{bmatrix}}q^{2(\varphi-\ell)(i-\ell)}\beta(n-i,\varphi-\ell)q^{m(n-i-(\varphi-\ell))}\nu^{[i](\ell)}(1,1;m)\\
    \overset{\eqref{equation:d'sandbetab's}}&{=} \sum_{\ell=0}^{\varphi}q^{2(\varphi-\ell)(i-\ell)}{\begin{bmatrix}\varphi\\\ell\end{bmatrix}}\beta(n-i,\varphi-\ell)q^{m(n-i-(\varphi-\ell))}\beta(i,i)\delta_{i\ell}\\
     \overset{\eqref{equation:betabstartdifferent}}&{=}{\begin{bmatrix}\varphi\\i\end{bmatrix}}{\begin{bmatrix}n-i\\ \varphi-i\end{bmatrix}}\beta(\varphi-i,\varphi-i)q^{m(n-\varphi)}\beta(i,i)\\
     \overset{\eqref{equation:betabstartsame}}&{=}{\begin{bmatrix}n-i\\ \varphi-i\end{bmatrix}}q^{m(n-\varphi)}\beta(\varphi,\varphi)
\end{align}

and so

\begin{align}
    \frac{1}{|\mathscr{C}^\perp|}\sum_{i=0}^n c_i'\psi_i(1,1) &  = \frac{1}{|\mathscr{C}^\perp|} \sum_{i=0}^{\varphi}c_i' {\begin{bmatrix}n-i\\\varphi-i\end{bmatrix}}q^{m(n-\varphi)}\beta(\varphi,\varphi) \\
    & = \beta(\varphi,\varphi)\frac{q^{m(n-\varphi)}}{|\mathscr{C}^\perp|}\sum_{i=0}^{\varphi} c_i'{\begin{bmatrix}n-i\\n-\varphi\end{bmatrix}}.        
\end{align}
Combining the results for each side, and simplifying, we finally obtain
\begin{equation}
    \sum_{i=0}^{n-\varphi}c_i {\begin{bmatrix}n-i\\\varphi\end{bmatrix}} = \frac{q^{m(n-\varphi)}}{|\mathscr{C}^\perp|}\sum_{i=0}^{\varphi} c_i'{\begin{bmatrix}n-i\\n-\varphi\end{bmatrix}}
\end{equation}
as required.
\end{proof}
\begin{note}
In particular, if $\varphi=0$ we have
\begin{equation}
    \sum_{i=0}^n c_i =\frac{q^{mn}}{|\mathscr{C}^\perp|}c_0' = \frac{q^{mn}}{|\mathscr{C}^\perp|}.
\end{equation}
In other words
\begin{equation}
    |\mathscr{C}||\mathscr{C}^\perp|=q^{mn}.
\end{equation}
We note that $mn=\frac{t(t-1)}{2}$ for skew-symmetric matrices and $q^{\frac{t(t-1)}{2}}$ is the number of skew-symmetric matrices of size $t\times t$. As such, this is the simple fact that the dimensions of a code and that of its dual add up to the dimension of the whole space they belong to.
\end{note}
We can simplify Proposition \ref{prop:momentsbderiv} if $\varphi$ is less than the minimum distance of the dual code.

\begin{cor}\label{corrollary:simplificationpropbderiv}
Let $d_{SR}'$ be the minimum skew rank distance of $\mathscr{C}^\perp$. If $0\leq \varphi < d_{SR}'$ then
\begin{equation}
    \sum_{i=0}^{n-\varphi}{\begin{bmatrix} n-i \\ \varphi\end{bmatrix}}c_i = \frac{1}{|\mathscr{C}^\perp|}q^{m(n-\varphi)} {\begin{bmatrix}n\\\varphi\end{bmatrix}}.
\end{equation}
\end{cor}

\begin{proof}
We have $c_0'=1$ and $c_1'=\ldots=c_\varphi'=0$.
\end{proof}

%-------------------------------------------------------------
\subsection{Moments derived from the Skew-$q^{-1}$-Derivative}

The next proposition relates the moments of the skew rank distribution of a linear code to those of it's dual, this time using the skew-$q^{-1}$-derivative of the MacWilliams identity for the skew rank metric. Before proceeding we first need the following two lemmas.

\begin{lem}\label{lemma:deltas}

Let $\delta(\lambda,\varphi,j)=\displaystyle\sum_{i=0}^{j}{\begin{bmatrix}j\\i\end{bmatrix}}(-1)^{i}q^{2\sigma_{i}}\gamma(\lambda-2i,\varphi)$. Then for all $\lambda\in\mathbb{R},\varphi,j\in\mathbb{Z}$,

\begin{equation}
    \delta(\lambda,\varphi,j) = \gamma(2\varphi,j)\gamma(\lambda-2j,\varphi-j)q^{j(\lambda-2j)}.
\end{equation}

\end{lem}

\begin{proof}
Initial case: $j=0$.

\begin{align}
    \delta(\lambda,\varphi,0) & = {\begin{bmatrix}0\\0\end{bmatrix}}(-1)^{0}q^{2\sigma_{0}}\gamma(\lambda,\varphi) = \gamma(\lambda,\varphi) = \gamma(2\varphi,0)\gamma(\lambda,\varphi)q^{0(\lambda)}.
\end{align}
So the initial case holds. Now assume the case is true for $j=\overline{\jmath}$ and consider the $\overline{\jmath}+1$ case.

\begin{align}
    \delta(\lambda,\varphi,\overline{\jmath}+1) & = \sum_{i=0}^{\overline{\jmath}+1}{
    \begin{bmatrix}\overline{\jmath}+1\\i\end{bmatrix}}(-1)^{i}q^{2\sigma_i}\gamma(\lambda-2i,\varphi)\\
    \overset{\eqref{equation:thing1}}&{=}\sum_{i=0}^{\overline{\jmath}+1}\left(q^{2i}{
    \begin{bmatrix}\overline{\jmath}\\i\end{bmatrix}}+{
    \begin{bmatrix}\overline{\jmath}\\i-1\end{bmatrix}}\right)(-1)^{i}q^{2\sigma_i}\gamma(\lambda-2i,\varphi)\\
    & = \sum_{i=0}^{\overline{\jmath}}{
    \begin{bmatrix}\overline{\jmath}\\i\end{bmatrix}}(-1)^{i}q^{2\sigma_i}q^{2i}\gamma(\lambda-2i,\varphi)+\sum_{i=0}^{\overline{\jmath}}{
    \begin{bmatrix}\overline{\jmath}\\i\end{bmatrix}}(-1)^{i+1}q^{2\sigma_{i+1}}\gamma(\lambda-2(i+1),\varphi)\\
    \overset{\eqref{equation:gamma2step}}&{=} \sum_{i=0}^{\overline{\jmath}}{
    \begin{bmatrix}\overline{\jmath}\\i\end{bmatrix}}(-1)^{i}q^{2i}q^{2\sigma_i}\left(q^{\lambda -2i}-1\right)q^{2(\varphi-1)}\gamma(\lambda-2i-2,\varphi-1)\\
    \overset{\eqref{equation:gamma1step}}&{~-}\sum_{i=0}^{\overline{\jmath}}{
    \begin{bmatrix}\overline{\jmath}\\i\end{bmatrix}}(-1)^{i}q^{2\sigma_{i+1}}\left(q^{\lambda -2i-2}-q^{2(\varphi-1)}\right)\gamma(\lambda-2i-2,\varphi-1)\\
    & = \sum_{i=0}^{\overline{\jmath}}{
    \begin{bmatrix}\overline{\jmath}\\i\end{bmatrix}}(-1)^{i}q^{2\sigma_{i+1}}\gamma(\lambda-2i-2,\varphi-1)q^{\lambda -2}\left(q^{2\varphi}-1\right)\\
    & = q^{\lambda -2}\left(q^{2\varphi}-1\right)\delta(\lambda-2,\varphi-1,\overline{\jmath})\\
    & = q^{\lambda -2}\left(q^{2\varphi}-1\right)\gamma(2(\varphi-1),\overline{\jmath})q^{\overline{\jmath}(\lambda-2\overline{\jmath}-2)}\gamma(\lambda-2-2\overline{\jmath},\varphi-1-\overline{\jmath})\\
    \overset{\eqref{equation:gamma2step}}&{=} q^{(\overline{\jmath}+1)(\lambda-2(\overline{\jmath}+1))}\gamma(2\varphi,\overline{\jmath}+1)\gamma(\lambda-2(\overline{\jmath}+1),\varphi-(\overline{\jmath}+1)).
\end{align}
as required. Hence by induction the lemma is proved.
\end{proof}

\begin{lem}\label{lemma:epsilons}
Let $\varepsilon(\Lambda,\varphi,i)=\displaystyle\sum_{\ell=0}^{i}{\begin{bmatrix}i\\\ell\end{bmatrix}}{\begin{bmatrix}\Lambda-i\\\varphi-\ell\end{bmatrix}}q^{2\ell(\Lambda-\varphi)}(-1)^{\ell}q^{2\sigma_{\ell}}\gamma(2(\varphi-\ell),i-\ell)$. Then for all $\Lambda\in\mathbb{R},\varphi,i\in\mathbb{Z}$,
\begin{equation}\label{equation:epsilons}
    \varepsilon(\Lambda,\varphi,i) = (-1)^{i}q^{2\sigma_{i}}{\begin{bmatrix}\Lambda-i\\\Lambda-\varphi\end{bmatrix}}.
\end{equation}
\end{lem}

\begin{proof}
Initial case $i=0$,
\begin{align}
    \varepsilon(\Lambda,\varphi,0) = {\begin{bmatrix}0\\ 0\end{bmatrix}}{\begin{bmatrix}\Lambda\\\varphi\end{bmatrix}}q^{0}(-1)^{0}q^{0}\gamma(2\varphi,0) & ={\begin{bmatrix}\Lambda \\\varphi\end{bmatrix}},\\
    (-1)^{0}q^{0}{\begin{bmatrix}\Lambda\\\Lambda-\varphi\end{bmatrix}} & = {\begin{bmatrix}\Lambda\\\varphi\end{bmatrix}}.
\end{align}
So the initial case holds. Now suppose the case is true when $i=\overline{\imath}$. Then

\begin{align}
    \varepsilon(\Lambda,\varphi,\overline{\imath}+1) & = \sum_{\ell=0}^{\overline{\imath}+1}{\begin{bmatrix}\overline{\imath}+1\\\ell\end{bmatrix}}{\begin{bmatrix}\Lambda-\overline{\imath}-1\\\varphi-\ell\end{bmatrix}}q^{2\ell(\Lambda-\varphi)}(-1)^{\ell}q^{2\sigma_{\ell}}\gamma(2(
    \varphi-\ell),\overline{\imath}+1-\ell)\\
    \overset{\eqref{equation:thing5}}&{=} \sum_{\ell=0}^{\overline{\imath}+1}{\begin{bmatrix}\overline{\imath}\\\ell\end{bmatrix}}{\begin{bmatrix}\Lambda-{\overline{\imath}}-1\\\varphi-\ell\end{bmatrix}}q^{2\ell(\Lambda-\varphi)}(-1)^{\ell}q^{2\sigma_{\ell}}\gamma(2(\varphi-\ell),\overline{\imath}+1-\ell)\\
    & ~ + \sum_{\ell=1}^{\overline{\imath}+1}q^{2(\overline{\imath}+1-\ell)}{\begin{bmatrix}\overline{\imath}\\\ell-1\end{bmatrix}}{\begin{bmatrix}\Lambda-\overline{\imath}-1\\\varphi-\ell\end{bmatrix}}q^{2\ell(\Lambda-\varphi)}(-1)^{\ell}q^{2\sigma_{\ell}}\gamma(2(\varphi-\ell),\overline{\imath}+1-\ell)\\
    & = A + B, \quad \text{say}.
\end{align}
Now
\begin{align}
    A & = \left(q^{2\varphi}-q^{2\overline{\imath}}\right)\sum_{\ell=0}^{\overline{\imath}}{\begin{bmatrix}\overline{\imath}\\\ell\end{bmatrix}}{\begin{bmatrix}\Lambda-{\overline{\imath }}-1\\\varphi-\ell\end{bmatrix}}q^{2\ell(\Lambda-1-\varphi)}(-1)^{\ell}q^{2\sigma_{\ell}}\gamma(2(\varphi-\ell),\overline{\imath} -\ell)\\
    & = \left(q^{2\varphi}-q^{2\overline{\imath}}\right)\varepsilon(\Lambda-1,\varphi,\overline{\imath})\\
    & = \left(q^{2\varphi}-q^{2\overline{\imath}}\right)(-1)^{\overline{\imath}}q^{2\sigma_{\overline{\imath}}}{\begin{bmatrix}\Lambda-\overline{\imath}-1\\ \Lambda-1-\varphi\end{bmatrix}}
.\end{align}
and
\begin{align}
    B & =\sum_{\ell=0}^{\overline{\imath}}q^{2(\overline{\imath}-\ell)}{\begin{bmatrix}\overline{\imath}\\\ell\end{bmatrix}}{\begin{bmatrix}\Lambda-1-\overline{\imath}\\\varphi-\ell-1\end{bmatrix}}q^{2(\ell+1)(\Lambda-\varphi)}(-1)^{\ell+1}q^{2\sigma_{\ell+1}}\gamma(2(\varphi-\ell-1),\overline{\imath}-\ell)\\
    & = -q^{2(\overline{\imath}+\Lambda-\varphi)}\sum_{\ell=0}^{\overline{\imath}}{\begin{bmatrix}\overline{\imath}\\\ell\end{bmatrix}}{\begin{bmatrix}\Lambda-1-\overline{\imath}\\\varphi-1-\ell\end{bmatrix}}q^{2\ell(\Lambda-\varphi)}(-1)^{\ell}q^{2\sigma_{\ell}}\gamma(2(\varphi-\ell-1),\overline{\imath}-\ell)\\
    & = -q^{2(\overline{\imath}+\Lambda-\varphi)}\varepsilon(\Lambda-1,\varphi-1,\overline{\imath})\\
    & = -q^{2(\overline{\imath}+\Lambda-\varphi)}(-1)^{\overline{\imath}}q^{2\sigma_{\overline{\imath}}}{\begin{bmatrix}\Lambda-1-\overline{\imath}\\ \Lambda-\varphi\end{bmatrix}}.
\end{align}
So
\begin{align}
    \varepsilon(\Lambda,\varphi,\overline{\imath}+1) & = A + B \\
    & = (-1)^{\overline{\imath}}q^{2\sigma_{\overline{\imath}}}\left\{\left(q^{2\varphi}-q^{2\overline{\imath}}\right){\begin{bmatrix}\Lambda-1-\overline{\imath}\\ \Lambda-1-\varphi\end{bmatrix}}-q^{2(\overline{\imath}+\Lambda-\varphi)}{\begin{bmatrix}\Lambda-1-\overline{\imath}\\ \Lambda-\varphi\end{bmatrix}}\right\}\\
    \overset{\eqref{equation:thing3}}&{=} (-1)^{\overline{\imath}+1}q^{2\sigma_{\overline{\imath}}}\left\{q^{2(\overline{\imath}+\Lambda-\varphi)}{\begin{bmatrix}\Lambda-1-\overline{\imath}\\ \Lambda-\varphi\end{bmatrix}}-\left(q^{2\varphi}-q^{2\overline{\imath}}\right)\frac{\left(q^{2(\Lambda-\varphi)}-1\right)}{\left(q^{2(\varphi-\overline{\imath})}-1\right)}{\begin{bmatrix}\Lambda-1-\overline{\imath}\\ \Lambda-\varphi\end{bmatrix}}\right\}\\
    & = (-1)^{\overline{\imath}+1}{\begin{bmatrix}\Lambda-(\overline{\imath}+1)\\ \Lambda-\varphi\end{bmatrix}}q^{2\sigma_{\overline{\imath}}}\left\{\frac{q^{2(\overline{\imath}+\Lambda-\varphi)}\left(q^{2(\varphi-\overline{\imath})}-1\right)-\left(q^{2\varphi}-q^{2\overline{\imath}}\right)\left(q^{2(\Lambda-\varphi)}-1\right)}{\left(q^{2(\varphi-\overline{\imath})}-1\right)}\right\}\\
    & = (-1)^{\overline{\imath}+1}{\begin{bmatrix}\Lambda-(\overline{\imath}+1)\\ \Lambda-\varphi\end{bmatrix}}q^{2\sigma_{\overline{\imath}}}q^{2\overline{\imath}}\frac{q^{2(\varphi-\overline{\imath})}-1}{q^{2(\varphi-\overline{\imath})}-1}\\
    & = (-1)^{\overline{\imath}+1}q^{2\sigma_{\overline{\imath}+1}}{\begin{bmatrix}\Lambda-(\overline{\imath}+1)\\ \Lambda-\varphi\end{bmatrix}}
\end{align}
as required.
\end{proof}
\begin{prop}\label{prop:momentsbminusderivative}
For $0 \le \varphi \le n$ and a linear code $\mathscr{C} \subseteq \mathscr{A}_{q,t}$ with dimension $k$ and its dual $\mathscr{C}^\perp$ with weight distributions ${\boldsymbol{c}}$ and ${\boldsymbol{ c'}}$, respectively we have
\begin{equation}
    \sum_{i=\varphi}^n q^{2\varphi(n-i)}{\begin{bmatrix}i\\\varphi\end{bmatrix}}c_i = q^{k-m\varphi}\sum_{i=0}^{\varphi}(-1)^{i}{q^{2\sigma_{i}}q^{2i(\varphi-i)}\begin{bmatrix}n-i\\n-\varphi\end{bmatrix}}\gamma(m-2i,\varphi-i)c_{i}'.
\end{equation}
\end{prop}

\begin{proof}
As per Proposition \ref{prop:momentsbderiv}, we apply Theorem \ref{mainthm1} to $\mathscr{C}^{\perp}$ to get

\begin{equation}
    W_{\mathscr{C}}^{SR}(X,Y) = \frac{1}{|\mathscr{C}^\perp|}\overline{W}_{\mathscr{C}^{\perp}}^{SR}\left(X+\left(q^m-1\right)Y,X-Y\right)
\end{equation}
or equivalently
\begin{align}
    \sum_{i=0}^n c_i Y^{i}X^{n-i} & = \frac{1}{|\mathscr{C}^\perp|}\sum_{i=0}^{n}c_i'\left(X-Y\right)^{[i]}\ast \left(X+\left(q^m-1\right)Y\right)^{[n-i]}\\
    & = \frac{1}{|\mathscr{C}^\perp|}\sum_{i=0}^n c_i' \nu^{[i]}(X,Y;m)\ast \mu^{[n-i]}(X,Y;m). \label{equation:skew-q-1}
\end{align}

For each side of Equation \eqref{equation:skew-q-1}, we shall apply the skew-$q^{-1}$-derivative $\varphi$ times and then evaluate at $X=Y=1$.

For the left hand side, we obtain
\begin{align}
    \label{equation:propmomentsbminusLHS}
    \left(\sum_{i=0}^n c_iY^{i}X^{n-i}\right)^{\{\varphi\}}  & = \sum_{i=\varphi}^n c_iq^{2\varphi(1-i)+2\sigma_{\varphi}}\beta(i,\varphi)Y^{i-\varphi}X^{n-i}\\
    \overset{\eqref{equation:betabstartdifferent}}&{=} \sum_{i=\varphi}^{n}c_i q^{2\varphi(1-i)+2\sigma_{\varphi}} {\begin{bmatrix}i\\\varphi\end{bmatrix}}\beta(\varphi,\varphi)Y^{i-\varphi}X^{n-i}.
\end{align}
Then using $X=Y=1$ gives
\begin{equation}
    \sum_{i=\varphi}^{n}c_i q^{2\varphi(1-i)+2\sigma_{\varphi}} {\begin{bmatrix}i\\\varphi\end{bmatrix}}\beta(\varphi,\varphi)Y^{i-\varphi}X^{n-i} = \sum_{i=\varphi}^n q^{2\varphi(1-i)+2\sigma_\varphi}\beta(\varphi,\varphi){\begin{bmatrix}i\\\varphi\end{bmatrix}} c_i.
\end{equation}

We now move on to the right hand side. For simplicity we shall write $\mu(X,Y;m)$ as $\mu(m)$ and similarly $\nu(X,Y;m)$ as $\nu(m)$. Now by using Theorem \ref{Liebnizbminusderiv},

 \begin{equation}\label{equation:propmomentsbminusRHS}
 \begin{split}
     \left(\frac{1}{|\mathscr{C}^\perp|}\sum_{i=0}^n c_i' \nu^{[i]}(m)\ast \mu^{[n-i]}(m)\right)^{\{\varphi\}} & = \frac{1}{|\mathscr{C}^\perp|}\sum_{i=0}^n c_i' \left(\sum_{\ell=0}^\varphi{\begin{bmatrix}\varphi\\\ell\end{bmatrix}}q^{2\ell(n-i-\varphi+\ell)}\nu^{[i]\{\ell\}}(m)\ast \mu^{[n-i]\{\varphi-\ell\}}(m-2\ell)\right)\\
     & = \frac{1}{|\mathscr{C}^\perp|}\sum_{i=0}^n c_i' \psi_i(m)
    \end{split}
 \end{equation}
say. Applying Lemma \ref{lemma:bminusderivlemma} we get
\begin{align}
    \psi_i(m) & = \sum_{\ell=0}^{\varphi} {\begin{bmatrix}\varphi\\l\end{bmatrix}}q^{2\ell(n-i-\varphi+\ell)}\left\{(-1)^\ell \beta(i,\ell)\nu^{[i-\ell]}(m)\right\}\\
    & \ast\left\{q^{-2\sigma_{\varphi-\ell}}\beta(n-i,\varphi-\ell)\gamma(m-2\ell,\varphi-\ell)\mu^{[n-i-\varphi+\ell]}(m-2\varphi)\right\}.
\end{align}
Now let
\begin{equation}
    \Psi(X,Y;m-2\varphi) = \nu^{[i-\ell]}(X,Y;m)\ast\gamma(m-2\ell,\varphi-\ell)\mu^{[n-i-\varphi+\ell]}(X,Y;m-2\varphi).
\end{equation}
Then we apply the skew-$q$-product and set $X=Y=1$ to get
\begin{align}
    \Psi(1,1;m-2\varphi) & = \sum_{u=0}^{n-\varphi}\left[\sum_{p=0}^u q^{2p(n-i-\varphi+\ell)}\nu_p^{[i-\ell]}(m)\gamma(m-2\ell-2p,\varphi-\ell)\mu^{[n-i-\varphi+\ell]}_{u-p}(m-2\varphi-2p)\right]\\
    & = \sum_{r=0}^{i-\ell}q^{2r(n-i-\varphi+\ell)}\nu_r^{[i-\ell]}(m)\gamma(m-2\ell-2r,\varphi-\ell)\left[ \sum_{t=0}^{n-i-\varphi+\ell}\mu_t^{[n-i-\varphi+\ell]}(m-2\varphi-2r)\right]\\
    \overset{\eqref{equation:producttosumgauss}}&{=} \sum_{r=0}^{i-\ell}q^{2r(n-i-\varphi+\ell)}q^{(m-2\varphi-2r)(n-i-\varphi+\ell)}\nu_r^{[i-\ell]}(m)\gamma(m-2\ell-2r,\varphi-\ell)\\
    & = q^{(m-2\varphi)(n-i-\varphi+\ell)}\sum_{r=0}^{i-\ell}(-1)^r q^{2\sigma_{r}}{\begin{bmatrix}i-\ell\\r\end{bmatrix}}\gamma(m-2\ell-2r,\varphi-\ell)\\
    & = q^{(m-2\varphi)(n-i-\varphi+\ell)}q^{(i-\ell)(m-2i)}\gamma(2(\varphi-\ell),i-\ell)\gamma(m-2i,\varphi-i)
\end{align}
by Lemma \ref{lemma:deltas}. Now using Lemma \ref{lemma:betabmanipulation} and noting that $q^{2\ell(n-i-\varphi+\ell)}q^{-2\sigma_{\varphi-\ell}}=q^{2\ell(n-i)}q^{-2\sigma_{\varphi}}q^{2\sigma_\ell}$ we get
\begin{align}
    \psi_i(1,1;m) & = \sum_{\ell=0}^{\varphi}(-1)^\ell{\begin{bmatrix}\varphi\\\ell\end{bmatrix}}q^{2\ell(n-i-\varphi+\ell)} q^{-2\sigma_{\varphi-\ell}}\beta(i,\ell)\beta(n-i,\varphi-\ell)\Psi(1,1;m-2\varphi)\\
    & = q^{-2\sigma_{\varphi}}\beta(\varphi,\varphi)\sum_{\ell=0}^{\varphi}(-1)^{\ell}q^{2\ell(n-i)}q^{2\sigma_{\ell}}{\begin{bmatrix}i\\\ell\end{bmatrix}}{\begin{bmatrix}n-i\\ \varphi-\ell\end{bmatrix}}\Psi(1,1;m-2\varphi).
\end{align}
Writing that
\begin{align}
    q^{-2\sigma_{\varphi}}q^{2\ell(n-i)}q^{(m-2\varphi)(n-\varphi-i+\ell)}q^{(i-\ell)(m-2i)} & = q^{2\sigma_{\varphi}}q^{2\varphi(1-n)}q^{m(n-\varphi)}q^{2\ell(n-\varphi)}q^{2i(\varphi-i)}\\
    & = q^{\theta}q^{2l(n-\varphi)}
\end{align}
we get
\begin{align}
    \psi_i(1,1;m) & = q^{\theta}\beta(\varphi,\varphi)\gamma(m-2i,\varphi-i)\sum_{\ell=0}^i(-1)^\ell q^{2\ell(n-\varphi)} q^{2\sigma_{\ell}}{\begin{bmatrix}i\\\ell\end{bmatrix}}{\begin{bmatrix}n-i\\\varphi-\ell\end{bmatrix}}\gamma(2(\varphi-\ell),i-\ell) \\
    \overset{\eqref{equation:epsilons}}&{=} (-1)^iq^{\theta}q^{2\sigma_{i}}\beta(\varphi,\varphi){\begin{bmatrix}n-i\\n-\varphi\end{bmatrix}}\gamma(m-2i,\varphi-i).
\end{align}

Combining both sides, we obtain
\begin{equation}
    \sum_{i=\varphi}^n q^{2\varphi(1-i)+2\sigma_\varphi}\beta(\varphi,\varphi){\begin{bmatrix}i\\\varphi\end{bmatrix}} c_i = \frac{1}{|\mathscr{C}^\perp|}\sum_{i=0}^n c_i'(-1)^iq^{\theta}q^{2\sigma_{i}}\beta(\varphi,\varphi){\begin{bmatrix}n-i\\n-\varphi\end{bmatrix}}\gamma(m-2i,\varphi-i).
\end{equation}
Thus 
\begin{equation}
    \sum_{i=\varphi}^n q^{2\varphi(n-i)} {\begin{bmatrix}i\\\varphi\end{bmatrix}}c_i  = \frac{q^{m(n-\varphi)}}{|\mathscr{C}^\perp|}\sum_{i=0}^\varphi  (-1)^i q^{2\sigma_i}q^{2i(\varphi-i)}{\begin{bmatrix}n-i\\n-\varphi\end{bmatrix}}\gamma(m-2i,\varphi-i)c_i'.
\end{equation}
Then if $\mathscr{C}$ has dimension $k$ we have
\begin{equation}
    |\mathscr{C}|=q^k, ~ |\mathscr{C}^\perp|=q^{mn-k},
\end{equation}
so
\begin{equation}
    \frac{q^{m(n-\varphi)}}{|\mathscr{C}^\perp|} = \frac{q^{m(n-\varphi)}}{q^{mn-k}} = q^{k-m\varphi}
\end{equation}
as required.
\end{proof}
We can simplify Proposition \ref{prop:momentsbminusderivative} if $\varphi$ is less than the minimum distance of the dual code. Also we can introduce the \textbf{\textit{diameter}}, $\varrho_{SR}'$, to be the maximum distance between any two codewords of the dual code and simplify Proposition \ref{prop:momentsbminusderivative} again.

\begin{cor}
If $0\leq \varphi < d_{SR}'$ then
\begin{equation}
    \sum_{i=\varphi}^{n}q^{2\varphi(n-i)}{\begin{bmatrix}i\\\varphi\end{bmatrix}}c_i = q^{k-m\varphi}{\begin{bmatrix}n\\\varphi\end{bmatrix}}\gamma(m,\varphi).
\end{equation}
For $\varrho_{SR}'<\varphi\leq n$ then
\begin{equation}
    \sum_{i=0}^\varphi (-1)^i q^{2\sigma_{i}}q^{2i(\varphi-i)}{\begin{bmatrix} n-i\\n-\varphi\end{bmatrix}}\gamma(m-2i,\varphi-i)c_i = 0.
\end{equation}
\end{cor}

\begin{proof}
First consider $0\leq\varphi< d_{SR}'$, then $c'_0=1$, $c_1'=\ldots=c_\varphi'=0$. Also since ${\begin{bmatrix}n\\n-\varphi
\end{bmatrix}}={\begin{bmatrix}n\\ \varphi
\end{bmatrix}}$ the statement holds. Now if $\varrho_{SR}'<\varphi\leq n$ then applying Proposition \ref{prop:momentsbminusderivative} to $\mathscr{C}^\perp$ gives
\begin{equation}
    \sum_{i=\varphi}^n   q^{2\varphi(n-i)}{\begin{bmatrix}i\\\varphi\end{bmatrix}}c_i' = q^{mn-k-m\varphi} \sum_{i=0}^\varphi(-1)^iq^{2\sigma_i}q^{2i(\varphi-i)}{\begin{bmatrix}n-i\\n-\varphi
\end{bmatrix}}\gamma(m-2i,\varphi-i)c_i.
\end{equation}
So using $c_\varphi'=\ldots=c_n'=0$ we get
\begin{equation}
    0 = \sum_{i=0}^{\varphi}(-1)^i q^{2\sigma_i}q^{2i(\varphi-i)}{\begin{bmatrix}n-i\\n-\varphi\end{bmatrix}}\gamma(m-2i,\varphi-i)c_i
\end{equation}
as required.
\end{proof}

\subsection{MSRD Codes}
As an application for the MacWilliams Identity, we can derive an alternative proof for the explicit coefficients of the skew rank weight distribution for MSRD codes to that in \cite[Theorem 4]{DelsarteAlternating}. This is analogous to the results for MRD codes presented in \cite[Proposition 9]{gadouleau2008macwilliams}.

Firstly a lemma that will be needed.
\begin{lem}\label{lemma:sequences}
If $a_0,a_1,\ldots,a_\ell$ and $b_0,b_1,\ldots,b_\ell$ are two sequences of real numbers and if
\begin{equation}
    a_j = \sum_{i=0}^{j}{\begin{bmatrix} \ell-i\\ \ell- j\end{bmatrix}}b_i
\end{equation}
for $0\leq j\leq \ell$, then

\begin{equation}
    b_i = \sum_{j=0}^i (-1)^{i-j}q^{2\sigma_{i-j}}{\begin{bmatrix}\ell-j\\ \ell-i \end{bmatrix}}a_j
\end{equation}
for $0\leq i\leq \ell$.
\end{lem}

\begin{proof}
This result uses the property of skew-$q$-nary Gaussian coefficients \cite[Equation 10]{DelsarteAlternating}, that
\begin{equation}
    \sum_{k=i}^j (-1)^{k-i}q^{2\sigma_{k-i}}{\begin{bmatrix}k\\i\end{bmatrix}}{\begin{bmatrix}j\\k\end{bmatrix}} = \delta_{ij}.
\end{equation}
Then for $0\leq i \leq \ell$, 

\begin{align}
    \sum_{j=0}^i (-1)^{i-j}q^{2\sigma_{i-j}}{\begin{bmatrix}\ell -j\\ \ell-i\end{bmatrix}}a_j & = \sum_{j=0}^i(-1)^{i-j}q^{2\sigma_{i-j}}{\begin{bmatrix}\ell-j\\ \ell-i\end{bmatrix}}\left(\sum_{k=0}^j {\begin{bmatrix}\ell-k \\\ell-j\end{bmatrix}}b_k\right)\\
    & = \sum_{k=0}^i \sum_{j=k}^i (-1)^{i-j}q^{2\sigma_{i-j}}
    {\begin{bmatrix}\ell-j\\ \ell-i\end{bmatrix}}{\begin{bmatrix}\ell-k\\\ell-j\end{bmatrix}}b_k\\
    & = \sum_{k=0}^i b_k \left(\sum_{s=\ell-i}^{\ell-k}(-1)^{i-\ell+s}q^{2\sigma_{i-\ell+s}}{\begin{bmatrix}s\\ \ell-i\end{bmatrix}}{\begin{bmatrix}\ell-k\\s\end{bmatrix}}\right)\\
    & = \sum_{k=0}^i b_k \delta_{ik}\\
    & = b_i
\end{align}
as required.
\end{proof}
\begin{prop}
Let $\mathscr{C}\subseteq\mathscr{A}_{q,t}$ be a linear MSRD code with weight distribution $\boldsymbol{c}$. Then we have $c_0=1$ and for $0\leq r \leq n-d_{SR}$ 
\begin{equation}
    c_{d_{SR}+r} = \sum_{i=0}^r (-1)^{r-i} q^{2\sigma_{r-i}}{\begin{bmatrix}d_{SR}+r\\d_{SR}+i\end{bmatrix}}{\begin{bmatrix}n\\d_{SR}+r\end{bmatrix}}\left( \frac{q^{m\left(d_{SR}+i\right)}}{|\mathscr{C}|}-1\right).
\end{equation}
\end{prop}

\begin{proof}

It can be seen that this is equivalent to \cite[(15)]{GabidulinTheory}.
Now from Corollary \ref{corrollary:simplificationpropbderiv} we have
\begin{equation}
    \sum_{i=0}^{n-\varphi}{\begin{bmatrix}n-i\\ \varphi\end{bmatrix}}c_i  = \frac{1}{|\mathscr{C}^\perp|}q^{m(n-\varphi)}{\begin{bmatrix}n\\\varphi\end{bmatrix}}
\end{equation}
for $0\leq \varphi < d_{SR}'$. Now if a linear code $\mathscr{C}$ is MSRD, with minimum distance $d_{SR}$ then $\mathscr{C}^\perp$ is also MSRD with minimum distance $d_{SR}'=n-d_{SR}+2$ \cite[p35]{DelsarteAlternating}. So Corollary \ref{corrollary:simplificationpropbderiv} holds for $0\leq \varphi\leq n-d_{SR}=d_{SR}'-2$. We therefore have $c_0=1$ and $c_1=c_2=\ldots=c_{d_{SR}-1}=0$ and setting $\varphi=n-d_{SR}-j$ for $0\leq j\leq n-d_{SR}$ we get 

\begin{align}
    {\begin{bmatrix}n\\n-d_{SR}-j\end{bmatrix}} + \sum_{i=d_{SR}}^{d_{SR}+j}{\begin{bmatrix}n-i\\n-d_{SR}-j\end{bmatrix}}c_i & = \frac{1}{|\mathscr{C}^\perp|}q^{m(d_{SR}+j)}{\begin{bmatrix}n\\n-d_{SR}-j\end{bmatrix}}\\
    \sum_{r=0}^j {\begin{bmatrix}n-d_{SR}-r\\n-d_{SR}-j\end{bmatrix}}c_{r+d_{SR}} & = {\begin{bmatrix}n\\n-d_{SR}-j\end{bmatrix}}\left(\frac{q^{m(d_{SR}+j)}}{|\mathscr{C}^\perp|}-1\right).
\end{align}
Applying Lemma \ref{lemma:sequences} with $\ell = n-d_{SR}$ and $b_r = c_{r+d_{SR}}$ then setting 
\begin{equation}
    a_j = {\begin{bmatrix}n\\n-d_{SR}-j\end{bmatrix}}\left(\frac{q^{m(d_{SR}+j)}}{|\mathscr{C}^\perp|}-1\right)
\end{equation}
gives
\begin{equation}
     \sum_{r=0}^j {\begin{bmatrix}n-d_{SR}-r\\n-d_{SR}-j\end{bmatrix}}b_r = a_j
\end{equation}
and so
\begin{align}
    b_r = c_{r+d_{SR}} & = \sum_{i=0}^r (-1)^{r-i}q^{2\sigma_{r-i}}{\begin{bmatrix}n-d_{SR}-i\\n-d_{SR}-r\end{bmatrix}}a_i\\
    & = \sum_{i=0}^r (-1)^{r-i}q^{2\sigma_{r-i}}{\begin{bmatrix}n-d_{SR}-i\\n-d_{SR}-r\end{bmatrix}}{\begin{bmatrix}n\\n-d_{SR}-i\end{bmatrix}}\left(\frac{q^{m(d_{SR}+i)}}{|\mathscr{
    C}^\perp|}-1\right).
\end{align}
But we have
\begin{align}
    {\begin{bmatrix}n-d_{SR}-i\\n-d_{SR}-r\end{bmatrix}}{\begin{bmatrix}n\\n-d_{SR}-i\end{bmatrix}} & = {\begin{bmatrix}n-(d_{SR}+i)\\n-(d_{SR}+r)\end{bmatrix}}{\begin{bmatrix}n\\d_{SR}+i\end{bmatrix}}\\
    \overset{\eqref{equation:producttosumgauss}}&{=} {\begin{bmatrix}d_{SR}+r\\d_{SR}+i\end{bmatrix}}{\begin{bmatrix}n\\n-(d_{SR}+r)\end{bmatrix}}\\
     \overset{\eqref{equation:gaussianswapplaces}}&{=} {\begin{bmatrix}d_{SR}+r\\d_{SR}+i\end{bmatrix}}{\begin{bmatrix}n\\d_{SR}+r\end{bmatrix}}.
\end{align}
Therefore
\begin{equation}
    c_{r+d_{SR}} = \sum_{i=0}^r (-1)^{r-i}q^{2\sigma_{r-i}}{\begin{bmatrix}d_{SR}+r\\d_{SR}+i\end{bmatrix}}{\begin{bmatrix}n\\d_{SR}+r\end{bmatrix}}\left(\frac{q^{m(d_{SR}+i)}}{|\mathscr{
    C}^\perp|}-1\right)
\end{equation}
as required.
\end{proof}

\begin{note}
We note again that $mn=\frac{t(t-1)}{2}$ for skew-symmetric matrices and $|\mathscr{C}||\mathscr{C}^\perp|=q^{mn}$ which can be can be used to simplify this to 
\begin{equation}
    c_{r+d_{SR}} = \sum_{i=0}^r (-1)^{r-i}q^{2\sigma_{r-i}}{\begin{bmatrix}d_{SR}+r\\d_{SR}+i\end{bmatrix}}{\begin{bmatrix}n\\d_{SR}+r\end{bmatrix}}\left(|\mathscr{C}|q^{m(d_{SR}+i-n)}-1\right).
\end{equation}

\end{note}

\newpage

\printbibliography

\end{document}